\documentclass[11pt,a4paper]{article}    
\usepackage{amsmath,amsfonts,amssymb,amsthm,amscd}
\usepackage[english]{babel}

\newtheorem{theorem}{Theorem}
\newtheorem{corollary}{Corollary}
\newtheorem{proposition}{Proposition}
\newtheorem{lemma}{Lemma}
\newtheorem{remark}{Remark}

\newcommand{\dsp}{\displaystyle}

\numberwithin{equation}{section}

\begin{document}
\title{The Vlasov-Poisson equation in $\mathbb{R}^3$\\ with infinite charge and velocities}
\author{Silvia Caprino*,  Guido Cavallaro$^+$ and Carlo Marchioro$^{++}$}
\maketitle
\begin{abstract}
We consider the Vlasov-Poisson equation  in $\mathbb{R}^3$ with initial data which are not $L^1$ in space and have unbounded support in the velocities. Assuming for the density a slight decay in space and a strong decay in velocities, we prove existence and uniqueness of the solution, thus generalizing the analogous result given in \cite{R3} for data compactly supported in the velocities.
\end{abstract}

\noindent \textit{Keywords}: Vlasov-Poisson  equation; infinitely  extended  plasma; unbounded velocities;  local energy.

\noindent
\textit{Mathematics  Subject  Classification 2010}: 82D10, 35Q99, 76X05

\footnotetext{*Dipartimento di Matematica Universit\`a Tor Vergata, via della Ricerca Scientifica, 00133 Roma (Italy), 
caprino@mat.uniroma2.it}
\footnotetext{$^+$Dipartimento di Matematica Universit\`a  La Sapienza, p.le A. Moro 2,  00185 Roma (Italy),  
cavallar@mat.uniroma1.it}
\footnotetext{$^{++}$International  Research Center M$\&$MOCS   (Mathematics and Mechanics of Complex Systems),
marchior@mat.uniroma1.it}
\section{Introduction}

We consider a one-species, positively charged plasma, under the influence of the auto-induced electric field. The equation describing the time evolution of this system is the following Vlasov-Poisson equation:
\begin{equation}
\label{Eq}
\left\{
\begin{aligned}
&\partial_t f(x,v,t) +v\cdot \nabla_x f(x,v,t)+  E(x,t) \cdot \nabla_v f(x,v,t)=0  \\
&E(x,t)=\int \frac{x-y}{|x-y|^3} \ \rho(y,t) \, dy     \\
&\rho(x,t)=\int f(x,v,t) \, dv \\
&f(x,v,0)=f_0(x,v)\geq 0
\end{aligned}
\right.
\end{equation} 
where $f(x,v,t)$ is the distribution of charged particles at the point of the phase space $(x,v)$ at time $t,$ $\rho$ is the spatial density and $E$ is the electric field.
Equation (\ref{Eq}) shows that $f$ is time-invariant along the solutions of the so called characteristics equations:
\begin{equation}
\label{ch}
\begin{cases}
\dsp  \dot{X}(t)= V(t)\\
\dsp  \dot{V}(t)= E(X(t),t) \\
\dsp ( X(0), V(0))=(x,v)
 \end{cases}
\end{equation}
where $$(X(t),V(t))=(X(x,v,t),V(x,v,t))$$ denote position and velocity of a particle starting at time $t=0$ from $(x,v)$
and the initial datum $(x,v)$ is distributed according to $f_0$. It is well known that along (\ref{ch})
the partial differential equation (\ref{Eq}) transforms into an ordinary differential equation, hence
a result of existence and uniqueness of solutions to (\ref{ch}) implies the same result for solutions to (\ref{Eq}),
with regularity properties depending on the regularity of $f_0$. Since $f$ is time-invariant along this motion,
$f(X(t), V(t),t)=f_0(x,v)$,
and 
\begin{equation}
\|f(t)\|_{L^\infty}=\|f_0\|_{L^\infty}.
\label{f0}
\end{equation}

This equation has been widely studied and the problem of existence and uniqueness of the solution in ${\mathbb{R}}^3$, for $L^1$ data, is completely solved, as it can be seen in many papers. We quote \cite{L, Pf, S} and a nice review of many such results  \cite{G} for all. The subsequent problem of a spatial density not belonging to $L^1$ has been investigated since many years, for instance in papers \cite{CCMP, CMP} and more recently in \cite{CCM2, Rem, R3}. Other papers related to this problem, with many different optics, are, to quote some of them,  \cite{Ch, J, P, S1, S2, S3}. In particular, in \cite{Rem} and  \cite{R3} we have considered an infinitely extended plasma, confined in a cylinder by an external magnetic field and in the whole space respectively. In both cases the spatial density is not supposed to be in $L^1$ and in both cases it is of primary relevance to achieve a good position of the equations, since the electric field could be infinite. To avoid this problem, we have assumed that the spatial density, even if not integrable, is slightly decaying at infinity. Another problem coming from the infinite charge of the plasma is that the total energy of the system is infinite. We bypass this difficulty by introducing the $local \ energy,$ that is a sort of energy of a bounded region, which however takes into account the whole interaction with the rest of the particles. By our hypotheses at $t=0$ it turns out that the local energy is finite and has good properties, such to enable us to prove, in both papers \cite{Rem} and  \cite{R3}, the existence and uniqueness of the solution globally in time, for initial data which are not $L^1$ in space. These results are proved in case that the initial distribution $f_0$ has compact support in the velocities.

In the present paper we extend the analysis to the unbounded velocities case, which appears physically more relevant.
The strategy of the proof is the following: we start from the already known case with a cutoff $N$ on the maximal velocity and we study the limit $N\to \infty$. Assuming a slight decay in space and a strong decay in velocity of the initial distribution $f_0(x,v)$,
we prove that the limit $N\to \infty$ does exist, satisfies the Vlasov-Poisson equation, preserves the decay behavior and
it is unique in this class.
We remark that this decay law includes the important Maxwell-Boltzmann distribution.

The main point for the present generalization is a sharp estimate on the electric field which refines the preceding
one given in \cite{R3}.   We observe that in the present paper the plasma can move in the whole space ${\mathbb{R}}^3$.   Recently an unbounded velocity case has been studied for a plasma confined in an unbounded cylinder by a 
magnetic mirror \cite{inf},  where the plasma moves in a quasi one-dimensional region, and this allow for a slighter spatial
decay at infinity.

\noindent  We refer to the characteristics equations (\ref{ch}) since, as it is well known, the existence of a unique solution to (\ref{ch}) would imply the same result for (\ref{Eq}), in case of smooth data.
The main results of the present paper are stated in the following Theorems:

\medskip

\begin{theorem}
Let us fix an arbitrary positive time T. Let $f_0$ satisfy the following hypotheses:
\begin{equation}
0\leq f_0(x,v)\leq C_0 e^{- \lambda |v|^2}g(|x|)\label{G}
 \end{equation}
 where $g$ is a positive, bounded, continuous function satisfying, for any $ i\in {\mathbb{Z}}^3\setminus{\{0\}},$
\begin{equation}
\int_{|i-x|\leq1}g(|x|)\ dx \leq C_1\frac{1}{|i|^{2+\epsilon}} 
\label{ass}
\end{equation}
for any fixed $1/15<\epsilon <1$,
 being $\lambda,$ $C_0$ and $C_1$  positive constants. Then there exists a solution to equations (\ref{ch}) on $[0, T]$ and positive constants $C$ and $\bar{\lambda}$ such that  
 \begin{equation}
 0\leq f(x,v,t)\leq Ce^{-\bar \lambda |v|^2},\label{dect}
 \end{equation}
and for any $ i\in {\mathbb{Z}}^3\setminus {\{0\}}$
\begin{equation}
\int_{|i-x|\leq1}\rho(x,t)\ dx \leq C\frac{\log^\frac32 (1+ |i|)}{|i|^{2+\epsilon}}.\label{asst}
\end{equation}
 \label{ent}
This solution is unique in the class of those satisfying (\ref{dect}) and (\ref{asst}).
\label{th_02}
\end{theorem}

\begin{remark}
The assumption (\ref{ass}) can be satisfied in case that the spatial density $\rho(x,0),$ even being not integrable, has a suitable decay at infinity,  but also whenever $\rho(x,0)$ is piecewise constant (or has an oscillatory character) over suitable sparse sets.
Hence hypothesis (\ref{ass}) allows for spatial densities which possibly do not belong  to any $L^p$ space.
If we assume that the spatial density is point-wise decreasing for large $|x|$, then Theorem \ref{th_02}
of course  remains valid, but the thesis can be improved, since we are able to show that at any time
$t\in [0,T]$ the spatial density keeps the same decreasing property. This is the object of the next Theorem 2.

Note that the upper bound for $\epsilon$ is due to allow infinite mass. We do not study the cases $\epsilon=1$
(a border case with infinite mass) and $\epsilon>1$ (finite mass),  in which the proof is simpler.

\end{remark}

\begin{theorem}
Let us fix an arbitrary positive time T. Let $f_0$ satisfy the following hypotheses:
\begin{equation}
0\leq f_0(x,v)\leq C_0 e^{- \lambda |v|^2}g(|x|)\label{Ga}
 \end{equation}
 where $g$ is a bounded, continuous, not increasing  function such that, for $|x|\geq 1$,
\begin{equation}
g(|x|) = C\frac{1}{|x|^{2+\epsilon}} 
\label{asp}
\end{equation}
for any fixed  $1/15<\epsilon <1$,
  being $\lambda,$ $C_0$ and $C$  positive constants. Then there exists a solution to equations (\ref{ch}) on $[0, T]$ and positive constants $C'$ and ${\lambda}'$ such that  
 \begin{equation}
 0\leq f(x,v,t)\leq C'e^{- \lambda' |v|^2}g(|x|)\label{dec}.
 \end{equation}
This solution is unique in the class of those satisfying (\ref{dec}).
\label{th_03}
\end{theorem}

\bigskip

The paper is devoted to the proofs of Theorems \ref{th_02} and \ref{th_03} and is planned in the following way. In Section 2 we define the partial dynamics, which is an evolution equation for a system regularized with respect to our purpose, that is with a cutoff in the velocities. For this system we introduce the local energy and we state many properties, in particular the bound on the auto-induced electric field and the bound on the local energy. The proof of Theorems \ref{th_02} and \ref{th_03} are given in Section 3, where we show that the estimates proved in the previous section can be made uniform with respect to the cutoff. Finally in Section 4 we give the proof of the main estimate on the electric field, which is fundamental in the proofs of both theorems, and the Appendix  is devoted to the proofs of some technical lemmas. 

\section{The partial dynamics }

Beside system (\ref{ch}) we introduce a modified differential system, called $the \\ partial \ dynamics,$ in which the initial density has compact support in the velocities. More precisely, for any positive integer $N,$  we  consider the following equations:
\begin{equation}
\label{chN}
\begin{cases}
\dsp  \dot{X}^N(t)= V^N(t)\\
\dsp  \dot{V}^N(t)= E^N(X^N(t),t) \\
\dsp ( X^N(0), V^N(0))=(x,v), 
 \end{cases}
\end{equation}
where
\begin{equation*}
\left(X^N(t),V^N(t)\right)=\left(X^N(x,v,t), V^N(x,v,t)\right)
\end{equation*} 
\begin{equation*}
E^N(x,t)=\int \frac{x-y}{|x-y|^3} \ \rho^N(y,t) \, dy
\end{equation*} 
\begin{equation*}
\rho^N(x,t)=\int f^N(x,v,t) \, dv
\end{equation*} 
and
\begin{equation*}
f^N(X^N(t),V^N(t),t)=f^N_0(x,v),
\end{equation*} 
being the initial distribution $f_0^N$ defined as
\begin{equation}
f^N_0(x,v)=f_0(x,v)\chi\left(|v|\leq N\right) \label{ic}
\end{equation}
where $\chi\left(\cdot\right)$ is the characteristic function of the set $\left(\cdot\right)$.  Since in the equations (\ref{chN}) the initial datum $f_0^N$ has compact support in the velocities, for any fixed $N$ the solution does exist unique, as it is proved in \cite{R3}. Our purpose, in order to prove both Theorems \ref{th_02} and \ref{th_03}, is to show that this solution converges, as $N\to\infty,$ to a solution to (\ref{ch}).

Along the paper some positive constants will appear, generally denoted by $C,$ except some which will be numbered in order to be quoted in the sequel.  All of them will depend exclusively on $ \|f_0\|_{L^\infty}$ and an arbitrarily fixed, once for ever, time $T,$ but not on $N,$ while any dependence on $N$ in the estimates will be clearly stressed.

\subsection{The local energy}

We define the $local\ energy$, already introduced in our previous papers, as a fundamental tool to deal with the infinite charge of the plasma.

 For any vector $\mu\in \mathbb{R}^3$ and any $R>0$ we define the function:
\begin{equation}
\varphi^{\mu,R}( x)=\varphi\Bigg(\frac{| x-\mu|}{R}\Bigg) \label{a}
\end{equation}
with $\varphi$ a smooth function such that:
\begin{equation}
\varphi(r)=1 \ \ \hbox{if} \ \ r\in[0,1] \label{b}
\end{equation}
\begin{equation}
\varphi(r)=0 \ \ \hbox{if} \ \ r\in[2,+\infty) \label{c}
\end{equation}
\begin{equation}
-2\leq \varphi'(r)\leq 0.\label{d}
\end{equation}
We define the local energy as:
\begin{equation}
\begin{split}
& W^N( \mu,R,t)=\\&\frac12 \int d x\, \varphi^{\mu,R}( x)\left[\int d v \ |v|^2 f^N( x, v,t)\,+
 \rho^N( x,t)\int dy \ 
\frac{ \rho^N ( y,t)}{  | x- y|}\right].
\end{split} \label{e}
\end{equation}
 It depends on the property of $f^N$ whether or not $W^N$ is bounded. For the moment, we stress that the local energy takes into account the complete interaction with the rest of the plasma out of the sphere of center $\mu$ and radius $2R.$

We set

\begin{equation}
Q^N(R,t)=\sup_{\mu\in {\mathbb {R}}^3 }W^N( \mu,R,t).\label{Q}
\end{equation}

First of all we observe that the hypotheses in Theorem \ref{th_02} (and hence those in Theorem \ref{th_03} which are more strict) ensure that the local energy is bounded at time $t=0$ and that the following holds:
\begin{lemma}
In the hypotheses of Theorem \ref{th_02}, $\forall  R\geq 1$ it holds
\begin{equation}
Q^N(R,0)\leq CR^{1-\epsilon}.
\end{equation}\label{luc}
\end{lemma}
\begin{proof}
We take any $\mu\in \mathbb{R}^3$ and $R\geq 1,$ assumed integer for simplicity, and we estimate $W^N(\mu, R, 0).$ By assumption (\ref{G}) we have
\begin{equation}
\begin{split}
W^N(\mu, R, 0)\leq C
 \int d x  \,\varphi^{\mu,R}( x)g(|x|)\left[1+\int dy \ 
\frac{ g( |y|)}{  | x- y|}\right].\label{C4}
\end{split}
\end{equation}
By the definition of the function $\varphi^{\mu,R}$ it is
\begin{equation}
\int \varphi^{\mu,R}( x)g(|x|)\,dx\leq \int_{|\mu-x|\leq 2R} g(|x|)\,dx.\label{fi}
\end{equation}
We estimate the integral on the right in both cases, $|\mu|\leq 3R$ and $|\mu|> 3R.$ Consider the case $|\mu|\leq 3R.$ By assumption (\ref{ass}) we have:
\begin{equation*}
\begin{split}
&\int_{|\mu-x|\leq 3R}g(|x|)\,dx \leq \int_{|x|\leq 5R}g(|x|)\,dx \leq \\& \int_{|x|\leq 1}g(|x|)\,dx+\sum_{\substack{i\in \mathbb{Z}^3\\1<|i|\leq 5 R}}\int_{ |i-x|\leq 1} g(|x|)\,dx\leq\\& C\left[1 + \sum_{\substack{i\in \mathbb{Z}^3\\1<|i|\leq 5 R}}\frac{1}{|i|^{2+\epsilon}}\right]\leq C  R^{1-\epsilon}.\end{split}
\end{equation*}
If on the contrary it is $|\mu|> 3R,$ then
\begin{equation*}
\begin{split}
\int_{|\mu-x|\leq 2R}g(|x|)\,dx \leq  \ &C\left[1+\sum_{\substack{i\in \mathbb{Z}^3\\1<|i|\leq2 R}}\int_{ |\mu+i-x|\leq1} g(|x|)\,dx\right]\leq \\& C\left[ 1+\sum_{\substack{i\in \mathbb{Z}^3\\1<|i|\leq2 R}}\frac{1}{|\mu+i|^{2+\epsilon}}\right]\end{split} \label{W11}
\end{equation*}
Since in this case it is $|\mu +i|\geq \frac{|i|}{2},$ we get 
\begin{equation}
\int_{|\mu-x|\leq2 R}g(|x|)\,dx \leq C\left[1+\sum_{\substack{i\in \mathbb{Z}^3\\1<|i|\leq2 R}}\frac{1}{|i|^{2+\epsilon}}\right]\leq C  R^{1-\epsilon}.
\end{equation}
Hence, in both cases we have 
\begin{equation*}
\int_{|\mu-x|\leq 2R}g(|x|)\,dx \leq C  R^{1-\epsilon}
\end{equation*}
which implies, by (\ref{fi}),
\begin{equation}
\int  \varphi^{\mu,R}( x)g(|x|)\,dx\leq C  R^{1-\epsilon}.\label{C2}
\end{equation}
Let us now estimate the potential energy of a single particle. 
\begin{equation}
\begin{split}
\int \frac{g(|y|)}{|x-y|}\,dy &\leq C\left[\int_{|x-y|\leq 1} \frac{g(|y|)}{|x-y|}\,dy+\sum_{i\in {\mathbb{Z}}^3:|i|\geq 1}\frac{1}{|i|}\int_{| x+i-y|\leq 1}g(|y|)\,dy\right]\\&\leq C\left[1+\sum_{\substack{i\in {\mathbb{Z}}^3:|i|>1\\|x+i|\geq 1}}\frac{1}{|i||x+i|^{2+\epsilon}}\right].
\end{split}\label{e3}
\end{equation}
Now, considering in the sum in (\ref{e3}) the two subsets of indices $\{i:|i|\leq |x+i|\}$ and $\{i:|i|> |x+i|\},$ we get 
$$
\sum_{\substack{i\in {\mathbb{Z}}^3:|i|>1\\|x+i|\geq 1}}\frac{1}{|i||x+i|^{2+\epsilon}}\leq \sum_{i\in {\mathbb{Z}}^3:|i|\geq 1}\frac{1}{|i|^{3+\epsilon}} + \sum_{i\in {\mathbb{Z}}^3:|x+i|\geq 1}\frac{1}{|x+i|^{3+\epsilon}}\leq C.
$$
Hence we have proved that
\begin{equation}
\int \frac{g(|y|)}{|x-y|}\,dy \leq C \label{C3}
\end{equation}
By (\ref{C4}), estimates (\ref{C2}) and (\ref{C3}) prove the thesis.
\end{proof}

We fix once for ever an arbitrary time $T>0,$ and all the estimates in the sequel are to be intended to hold on the interval $[0,T].$

 Let us introduce the following functions, for any $t\in [0,T]:$ 
\begin{equation}
{\mathcal{V}}^N(t)=\max
\left\{ \widetilde C,  \sup_{s\in [0,t]}\sup_{(x,v)\in \mathbb{R}^3 \times B(0,N)}|V^N(x,v,s)| \right\} \label{mv}
\end{equation}
\begin{equation}
R^N(t)= 1 +\int_0^t {\mathcal{V}}^N(s)ds
\label{RN}
\end{equation}
where $\widetilde{C}>1$ is a positive constant chosen suitably large for further technical
purposes (see for example  Lemma \ref{lemv3l} in the Appendix), 
and  $B(0,N)$ is the ball in $\mathbb{R}^3$ of center $0$ and radius $N.$

\medskip

The following statement is our main result on the local energy. Its proof has already been done in previous works, as for instance \cite{R3}, and we will not repeat it here.  We only stress that it is strongly based on the positivity of the local potential energy.
\begin{proposition}
For any $ t\in [0,T]$ it holds
\begin{equation*}
Q^N(R^N(t), t)\leq CQ^N(R^N(t),0).
\end{equation*}
\label{propo}\end{proposition}

\begin{remark}
In the hypotheses of Theorems \ref{th_02} and \ref{th_03} the local energy, defined in (\ref{e}),
satisfies  at time $t$,  in the limit $N\to\infty$, the same properties which hold initially, that is it remains bounded  and it 
grows with $R$ as stated in  Lemma \ref{luc}.
\end{remark}

\subsection{Some preliminary estimates}

We give now some results on system (\ref{chN}).  Hereafter we assume that hypotheses of Theorem \ref{th_02} are satisfied, since the same results will follow from the hypotheses of Theorem \ref{th_03}, by an {\textit{a fortiori}} argument.

The use of the local energy and its properties give us the following estimate:
\begin{lemma}
For any $\mu\in \mathbb{R}^3$ and $R\geq 1$ it holds
\begin{equation*}
\int_{1\leq |\mu-x|\leq R}\frac{\rho^N(x,t)}{|\mu-x|^2}\,dx\leq CQ^N(R,t)^{\frac35}.
\end{equation*}
\label{L1}
\end{lemma}

\begin{proof}

Notice that it is, for any $a>0$ 
\begin{equation*}
\begin{split}
\rho^N(x,t)=&\int  f^N(x,v,t)\,dv\leq \int_{|v|\leq a}  f^N(x,v,t)\,dv\,+\\&\frac{1}{a^2}\int_{|v|>a}\ |v|^2 f^N(x,v,t)\,dv
\end{split}
\end{equation*}
so that, by (\ref{f0}), we have 
\begin{equation}
\rho^N(x,t)\leq C a^3+\frac{1}{a^2}k(x,t)
\end{equation}
where
 $$
 k(x,t)= \int  |v|^2 f^N(x,v,t)\,dv.
 $$
 By minimizing over $a$ and  taking the power $\frac53$ of both members we get 
\begin{equation*}
\rho^N(x,t)^{\frac53} \leq Ck(x,t)
\end{equation*}
which implies,  by integrating over the set $\{x: |\mu-x|\leq R\}$ 
\begin{equation}
\int_{|\mu-x|\leq R}\rho^N(x,t)^{\frac53} \,dx\leq CW^N(\mu, R, t)\leq CQ^N(R,t).\label{ss}
\end{equation}
Notice that we have bounded the local kinetic energy by $W^N,$ which is allowed by the fact that the local potential energy is positive. 
Now we have
\begin{equation}
 \begin{split}
&\int_{1\leq |\mu-x|\leq R}\frac{\rho^N(x,t)}{|\mu-x|^2}\,dx\leq\\& \left(\int_{|\mu-x|\leq R}  \  \rho^N(x,t)^{\frac53}\,dx\right)^{\frac35}\left(\int_{ 1\leq |\mu-x|} \frac{1}{|\mu-x|^{5}}\,dx\right)^{\frac25}
\end{split}
\end{equation} 
and (\ref{ss}) implies
\begin{equation}
 \begin{split}
&\int_{1\leq |\mu-x|\leq R}\frac{\rho^N(x,t)}{|\mu-x|^2}\,dx\leq   CQ^N(R,t)^{\frac35}.\end{split}\label{s'}
\end{equation} 
\end{proof}
\begin{lemma}
Let $R^N(t)$ be defined in (\ref{RN}). Then for any $\mu\in \mathbb{R}^3$ it holds
  \begin{equation}
\int_{3R^N(t)\leq |\mu-x| }\frac{\rho^N(x,t)}{|\mu-x|^2}\,dx\leq C.
\end{equation}\label{L1'}
\end{lemma}
\begin{proof}

We notice that if $|\mu-x|\geq 3R^N(t)$ then
\begin{equation}
\left|\mu-X^N(t)\right|\geq |\mu-x|-R^N(t)\geq 2R^N(t)
\end{equation}
and also
\begin{equation}
\left|\mu-X^N(t)\right|\geq |\mu-x|-\frac{|\mu-x|}{3}= \frac{2}{3}|\mu-x|.\label{var}
\end{equation}
 By the invariance of $f$ along the characteristics and by the change of variables $ \left(X^N(x,v,t), V^N(x,v,t)\right)\to (x,v)$ we get
\begin{equation}
\begin{split}
\int_{3R^N(t)\leq |\mu-x|} \frac{\rho^N(x,t)}{|\mu-x|^2}\,dy&=\int_{3R^N(t)\leq |\mu-x|} \frac{f^N(x,v,t)}{|\mu-x|^2}\,dxdv\leq \\& \int_{2R^N(t)\leq |\mu-x|} \frac{f_0^N({x},{v})}{|\mu-X^N(t)|^2}\,dxdv \end{split}\label{kk}
\end{equation}
so that, by (\ref{var}) and using the assumption (\ref{G}), we have 
\begin{equation}
\begin{split}
\int_{3R^N(t)\leq |\mu-x|} \frac{\rho^N(x,t)}{|\mu-x|^2}\,dx&\leq \frac94\int_{2R^N(t)\leq |\mu-x|} \frac{f_0^N(x,v)}{|\mu-x|^2}\,dxdv\leq\\& C\int_{2R^N(t)\leq |\mu-x|} \frac{g(|x|)}{|\mu-x|^2}\,dx.
\end{split}
\end{equation}

Now it is
\begin{equation*}
\begin{split}
&\int_{ 2R^N(t)\leq |\mu-x|} \frac{g(|x|)}{ |\mu-x|^2}\,dx
\leq \sum_{\substack{i\in \mathbb{Z}^3\setminus \{0\}\\ 2R^N(t)\leq |\mu-i|}}\int_{ |i-x|\leq 1} \frac{ g(|x|)}{ |\mu-x|^2}\,dx.
\end{split}
\end{equation*}
Since 
$$
|\mu-x|\geq |\mu-i| - 1\geq |\mu-i|- R^N(t)\geq\frac{|\mu-i|}{2}
$$
from hypothesis (\ref{ass}) it follows
\begin{equation}
\begin{split}
&\int_{ 2R^N(t)\leq |\mu-x|} \frac{g(|x|)}{ |\mu-x|^2}\,dx\leq\\&
\sum_{\substack{i\in \mathbb{Z}^3\setminus \{0\}\\ |\mu-i|\geq 1}}\frac{4}{|\mu-i|^2}\int_{ |i-x|\leq 1}  g(|x|)\,dx\leq 4C_1\sum_{\substack{i\in \mathbb{Z}^3\setminus \{0\}\\ |\mu-i|\geq 1}}\frac{1}{|\mu-i|^2|i|^{2+\epsilon}}.
\end{split}\label{g2}
\end{equation}
By considering in the sum the two subsets of indices $\{i: |\mu-i|\geq |i|\}$ and $\{i:|\mu-i|< |i|\},$ we get:
\begin{equation}
\begin{split}
&\sum_{\substack{i\in \mathbb{Z}^3:|i|\geq1\\ |\mu-i|\geq 1}}\frac{1}{|\mu-i|^2|i|^{2+\epsilon}}\leq \sum_{\substack{i\in \mathbb{Z}^3:\\|i|\geq1}}\frac{1}{|i|^{4+\epsilon}}+ \sum_{\substack{i\in \mathbb{Z}^3:\\  |\mu-i|\geq 1}}\frac{1}{|\mu-i|^{4+\epsilon}}\leq C
\end{split}\label{s''}
\end{equation}
which proves the thesis.
\end{proof}

Now we give a first bound on the electric field $E.$  
\begin{proposition}

\begin{equation}
\|E^N(t)\|_{L^\infty}\leq C_2{\mathcal{V}}^N(t)^{\frac43}Q^N(R^N(t),t)^{\frac13}.
\label{C1}
\end{equation}
\label{prop2}\end{proposition}
\begin{proof}

We have: 
\begin{equation}
|E^N(x,t)|\leq \sum_{k=1}^{3}{\mathcal{J}}_k\label{E3}
\end{equation}
where
$$
{\mathcal{J}}_1=\int_{0< |x-y|\leq a}  \frac{ \rho^N(y,t)}{|x-y|^2}\,dy
$$
with $a<1$ to be suitably chosen in what follows,  
$$
{\mathcal{J}}_2=\int_{a< |x-y|\leq 3R^N(t)}  \frac{ \rho^N(y,t)}{|x-y|^2}\,dy
$$
and
$$
{\mathcal{J}}_3=\int_{|x-y|> 3R^N(t)}  \frac{ \rho^N(y,t)}{|x-y|^2}\,dy.
$$
Let us estimate ${\mathcal{J}}_1$. We have:

\begin{equation*} 
{\mathcal{J}}_1\leq 4\pi \|\rho^N(t)\|_{L^{\infty}}a.
\end{equation*}
Since 
\begin{equation}
\rho^N(x,t)=\int dv\ f^N(x,v,t)\leq C {\mathcal{V}}^N(t)^{3}\label{dens}
\end{equation}
we get
\begin{equation}
{\mathcal{J}}_1\leq  C{\mathcal{V}}^N(t)^{3}a\label{A1}.
\end{equation}

By estimate (\ref{ss}) we get:
\begin{equation}
\begin{split}
{\mathcal{J}}_2(x,t)\leq C&\left(\int_{|x-y|\leq 3R^N(t)}  \  \rho^N(y,t)^{\frac53}\,dy\right)^{\frac35}\left(\int_{ a< |x-y|\leq 3R^N(t)} \frac{1}{|x-y|^5}\,dy\right)^{\frac25} \\&\leq CQ^N(3R^N(t),t)^{\frac35}\left[ a^{-\frac45}+R^N(t)^{-\frac45}\right].  \end{split}\label{i2}
\end{equation}
Now, it is easy to see that
\begin{equation}
Q^N(3R^N(t),t)\leq C Q^N(R^N(t),t).\label{e'}
\end{equation}
Indeed, recalling the definition of the function ${\mathcal{\varphi}}^{\mu,R}$ in the local energy (\ref{e}), for any positive number $R'>R$ 
and for any $\mu \in \mathbb{R}^3$ it is
$$
{\mathcal{\varphi}}^{\mu,R'}(x)={\mathcal{\varphi}}\left(\frac{|x-\mu|}{R'}\right)\leq \sum_{i\in {\mathbb{Z}^3}:|i|\leq \frac{R'}{R}} {\mathcal{\varphi}}\left(\frac{|x-(\mu+iR)|}{R}\right).
$$
Hence, since both terms in the function $W^N$ are positive, by monotony we have:
\begin{equation}
W^N(\mu, R',t)\leq  \sum_{i\in {\mathbb{Z}^3}:|i|\leq  \frac{R'}{R}} W^N(\mu+iR,R,t)\leq  \left(\frac{R'}{R}\right)^3Q^N(R,t).
\label{R'/R}
\end{equation}
This implies (\ref{e'}) which, used in (\ref{i2}), gives
\begin{equation}
{\mathcal{J}}_2(x,t)\leq CQ^N(R^N(t),t)^{\frac35} a^{-\frac45}.
\end{equation}

 The minimum value of ${\mathcal{J}}_1(x,t)+{\mathcal{J}}_2(x,t)$ is attained at 
$$
a=C {\mathcal{V}}^N(t)^{-\frac53}Q^N(R^N(t),t)^{\frac13}
$$
so that we get
\begin{equation}
{\mathcal{J}}_1(x,t)+{\mathcal{J}}_2(x,t)\leq  C {\mathcal{V}}^N(t)^{\frac43}Q^N(R^N(t),t)^{\frac13}. \label{E1}
\end{equation}
Finally by Lemma \ref{L1'} we have 
\begin{equation}
{\mathcal{J}}_3(x,t)\leq C.\label{E4}
\end{equation}

By (\ref{E3}), formulas  (\ref{E1}) and (\ref{E4}) imply the thesis.

\end{proof}

The main estimate on $E^N$ is stated in the following proposition. Its proof, rather lengthy, is given in Section 4. 
\begin{proposition} There exists a positive constant $C_3$ such that, for any $t\in [0,T]:$
$$
\int_0^t|E^N(X^N(s),s)|ds\leq C_3 \left[{\mathcal{V}}^N(t)\right]^\alpha\quad \quad
\forall \,\, \alpha \in  \left(\frac{5-\epsilon}{9}, \frac23\right)
$$
being $\epsilon$ the one in (\ref{ass}).
\label{propE}
\end{proposition}
This bound is finer than the analogous one in \cite{R3}, in which it was $\alpha <1.$  The presence of plasma particles having infinite velocities forces us to get a better estimate on the electric field,
hence the upper bound $2/3$ for the parameter $\alpha$ is a key tool in the proof of the convergence of the
iterative scheme in Section 3 (see (\ref{fin})), while the lower bound $(5-\epsilon)/9$ is used in eq. (\ref{b1}). 


As a consequence of this result we have the following
\begin{corollary}
\begin{equation}
{\mathcal{V}}^N(T) \leq C_4 N
\label{V^N}
\end{equation}
\begin{equation}
\rho^N(x, t) \leq C N^{3\alpha}
\label{rho_t}
\end{equation}
\begin{equation}
Q^N(R^N(t), t)\leq CN^{1-\epsilon}. 
\label{Q^N}
\end{equation}
\label{coro}
\end{corollary}
\begin{proof}
The bound on ${\mathcal{V}}^N(T)$ is obvious, and comes from Proposition \ref{propE}, being
\begin{equation}
|V^N(t)|\leq |v|+\int_0^t|E^N(X^N(s),s)|\ ds \label{vel}
\end{equation}
and recalling that $|v|\leq N$ and $\alpha <1.$

Now we prove (\ref{rho_t}). Setting
 $$
 (\bar{x}, \bar{v})=\left(X^N(x,v,t),V^N(x,v,t)\right)$$
by the invariance of the density  $f^N$ along the characteristics we have 
\begin{equation*}
\rho^N(\bar{x}, t) =\int f^N(\bar{x}, \bar{v},t)\ d\bar{v}=\int f_0^N(x,v)\ d\bar{v}.
\end{equation*}
We consider the set $A_1=\{\bar{v}:|\bar{v}|\leq 2C_3(C_4 N)^\alpha\}$ and its complementary set $A_1^c.$ By assumption (\ref{G}) we have:
\begin{equation}
\begin{split}
\rho^N(\bar{x}, t)&\leq\int_{A_1} f_0^N(x,v)\ d\bar{v}+C\int_{A_1^c} e^{-\lambda |v|^2}\ d\bar{v}\\& \leq   CN^{3\alpha}+  C\int_{A_1^c} e^{-\lambda |v|^2}\ d\bar{v}.   \label{f_0}
\end{split}
\end{equation}
 By (\ref{vel}), Proposition \ref{propE} and (\ref{V^N})  it is
$$
|v|\geq |\bar{v}|- C_3(C_4N)^\alpha
$$
and hence for any $\bar{v}\in A_1^c$ it is $|v|\geq \frac{|\bar{v}|}{2}.$
\noindent Hence from (\ref{f_0}) it follows
\begin{equation}
\begin{split}
&\rho^N(\bar{x}, t)\leq CN^{3\alpha}+C \int e^{-\frac{\lambda}{4}|\bar{v}|^2}\ d\bar{v} \leq CN^{3\alpha}.\end{split}
\end{equation}
Finally (\ref{Q^N}) comes from Proposition \ref{propo}, Lemma \ref{luc}, definition (\ref{RN}) and (\ref{V^N}).
 \end{proof}

\subsection{Estimate of the term $|X^N(t)-X^{N+1}(t)|$}

 Let $f_0$ satisfy the assumptions of Theorem \ref{th_02} and let us fix an initial condition $(x,v)$ in the support of $f_0^N.$ We consider the time evolved of the point in the phase space $(x,v)$ in the $N$-th and in the $(N+1)$-th dynamics, that is we consider $\left(X^N(t), V^N(t)\right)$ and $\left(X^{N+1}(t), V^{N+1}(t)\right),$ the solutions to eq. (\ref{chN}) with initial condition $f_0^N$ and $f_0^{N+1}$ respectively. We set
 $$
\delta^N(t)\, = \sup_{(x,v)\in \mathbb{R}^3 \times B(0,N)} \delta^N(x,v,t),
$$
$$
\delta^N(x,v,t)=  |X^N(t)-X^{N+1}(t)|.
$$

\noindent We prove the following result:
\begin{proposition}
For any $t\in [0,T]$ it holds
\begin{equation}
\delta^N(t)\leq  C\int_0^tdt_1\int_o^{t_1}dt_2 \left[ N^{3\alpha}\delta^N(t_2)\left(\left| \log \delta^N(t_2) \right|+1\right)+e^{-\frac{\lambda}{2} N^2}\right].\label{estN}
\end{equation}\label{p4}
\end{proposition}

\begin{proof}
We have
\begin{equation}
\begin{split}
&\delta^N(x,v,t) =\\&  \, \left| \int_0^t dt_1 \int_0^{t_1} dt_2 \, \Big[E^N\left(X^N(t_2), t_2\right)
 -  E^{N+1}\left(X^{N+1}(t_2), t_2\right)\Big] \right|\leq   \\
&\int_0^t dt_1 \int_0^{t_1} dt_2 \, \left[ \mathcal{F}_1(X^N(t_2), X^{N+1}(t_2),t_2) + \mathcal{F}_2(X^N(t_2), X^{N+1}(t_2),t_2) \right],
\end{split}
\label{iter_}
\end{equation}
where
\begin{equation}
\mathcal{F}_1(X^N(t), X^{N+1}(t),t) = \left|E^N\left(X^N(t), t\right)
-E^N\left(X^{N+1}(t), t\right)\right| 
\end{equation}
and
\begin{equation}
\mathcal{F}_2(X^{N+1}(t),t) = \left|E^N\left(X^{N+1}(t), t\right)
-E^{N+1}\left(X^{N+1}(t), t\right)\right|.
\end{equation}

\bigskip

\noindent
{\textbf{Estimate of the term}}  $\mathcal{F}_1$ 

\medskip

We prove a quasi-Lipschitz property for $E^N.$ We consider the difference $|E^N(x,t)-E^N(y,t)|$ at two generic points $x$ and $y,$ and set: 
\medskip

\begin{equation*}
\mathcal{F}_1'(x,y,t)= |E^N(x,t)-E^N(y,t)|\chi (|x-y|\geq 1)
\end{equation*}
and
\begin{equation*}
\mathcal{F}_1''(x,y,t)= |E^N(x,t)-E^N(y,t)|\chi (|x-y|< 1)
\end{equation*}
with $\chi$ the characteristic function. Hence it is
\begin{equation}
\mathcal{F}_1(X^N(t), X^{N+1}(t),t)=\mathcal{F}_1'(X^N(t), X^{N+1}(t),t)+\mathcal{F}_1''(X^N(t), X^{N+1}(t),t). \label{F}\end{equation}
\medskip
By Proposition \ref{prop2} and Corollary \ref{coro} we have
\begin{equation}
\begin{split}
\mathcal{F}_1'(x,y,t)\leq&\ |E^N(x,t)|+|E^N(y,t)|\leq C N^{\frac{5-\epsilon}{3}}
\end{split}\label{b0}
\end{equation}
By the range of the parameter $\alpha$ it is  $\frac{5-\epsilon}{3}<3\alpha,$ so that we get
\begin{equation}
\begin{split}
\mathcal{F}_1'(x,y,t)\leq CN^{3\alpha}\leq CN^{3\alpha}|x-y|. \label{b1}
\end{split}
\end{equation}

\medskip
Let us now estimate the term $\mathcal{F}_1''(x,y,t).$ 

We put $\zeta=\frac{x+y}{2}$ and consider the sets: 

\begin{equation}
\begin{split}
&S_1=\{z:|\zeta-z|\leq 2|x-y|\}, \\& S_2=\Big\{z:2|x-y|\leq |\zeta-z|\leq\frac{4}{|x-y|}\Big\}\\& S_3=\Big\{z:|\zeta-z|\geq \frac{4}{|x-y|}\Big\}.\label{S'}
\end{split}
\end{equation}
We have:
\begin{equation}
\begin{split}
&\mathcal{F}_1''(x,y,t)\leq  \int_{S_1\cup S_2\cup S_3} \left|\frac{1}{|x-z|^2}-\frac{1}{|y-z|^2}\right| \rho^N(z,t) \,dz.
\end{split} \label{s}
\end{equation}

By (\ref{rho_t}) and the definition of $\zeta$ we have
\begin{equation}
\begin{split}
&\int_{S_1} \left|\frac{1}{|x-z|^2}-\frac{1}{|y-z|^2}\right| \rho^N(z,t) \,dz\leq\\& CN^{3\alpha}\int_{S_1} \left(\frac{1}{|x-z|^2}+\frac{1}{|y-z|^2}  \right)\,dz\leq C N^{3\alpha} |x-y|.
\end{split}
\label{Ai1}\end{equation}
Let us pass to the integral over the set $S_2.$ By the Lagrange theorem, there exists $\xi_z$ such that $\xi_z=\theta x +(1-\theta)y$ and $\theta\in [0,1]$ (depending on $z$), for which
\begin{equation} 
\begin{split}
\int_{S_2} \left|\frac{1}{|x-z|^2}-\frac{1}{|y-z|^2}\right|& \rho^N(z,t) \, dz\leq C |x-y|\int_{S_2} \frac{\rho^N(z,t)}{|\xi_z-z|^3}\,dz\\&\leq  CN^{3\alpha}|x-y|\int_{S_2} \frac{1}{|\xi_z-z|^3}\,dz
\end{split}\label{D'}
\end{equation}
Since in $S_2$ it results $|\xi_z-z|\geq \frac{|\zeta -z|}{2}$, we have
\begin{equation}
\begin{split}
&\int_{S_2} \frac{1}{|\xi_z-z|^3} \,dz\leq 8 \int_{S_2} \frac{1}{|\zeta-z|^3} \,dz \leq C  (|\log |x-y||+1) 
\end{split}
\end{equation}
and 
\begin{equation}
\int_{S_2} \left|\frac{1}{|x-z|^2}-\frac{1}{|y-z|^2}\right|  \rho^N(z,t) \, dz \leq C N^{3\alpha} |x-y| \, (|\log |x-y||+1).
\label{D}
\end{equation}

For the last integral over $S_3,$ again by the Lagrange theorem, we have for some $\xi_z=\theta x +(1-\theta)y$ and $\theta\in [0,1],$  
\begin{equation}
 \int_{ S_3} \left|\frac{1}{|x-z|^2}-\frac{1}{|y-z|^2}\right| \rho^N(z,t)\, dz\leq C |x-y| \int_{ S_3} \frac{\rho^N(z,t)}{|\xi_z-z|^3}\,dz.\label{S}
 \end{equation}
Notice that, since $z\in S_3$ and $|x-y|<1$ by definition of ${\mathcal{F}}_1''$,  it is 
$$
|\xi_z - z|\geq  |\zeta - z| - |x-y|,
$$
and
$$
\frac{1}{|\xi_z - z|} \leq \frac{1}{|\zeta - z| - |x-y|} = \frac{1}{\frac34|\zeta - z| +\frac14 |\zeta-z|- |x-y|}
\leq \frac43 \frac{1}{|\zeta-z|},
$$
 then we have
\begin{equation}
\begin{split}
\int_{ S_3} \frac{\rho^N(z,t)}{|\xi_z-z|^3}\,dz \leq C\int_{|\zeta-z|\geq 4 } \frac{\rho^N(z,t)}{|\zeta-z|^3}\,dz.
\end{split}
\end{equation}
Lemmas  \ref{L1}, \ref{L1'} and (\ref{e'}) imply
\begin{equation}
\int_{ S_3} \frac{\rho^N(z,t)}{|\xi_z-z|^3}\,dz\leq C\int_{|\zeta-z|\geq 4 } \frac{\rho^N(z,t)}{|\zeta-z|^3}\,dz \leq CQ^N(R^N(t),t)^{\frac35} + C
\end{equation}
and by (\ref{Q^N}) we get
\begin{equation}
\int_{S_3} \frac{\rho^N(z,t)}{|\xi_z-z|^3}\,dz\leq CN^{\frac35(1-\epsilon)}.\label{S2}
\end{equation}
Using this estimate in (\ref{S}) and going back to (\ref{s}), by (\ref{Ai1}) and (\ref{D}) we have proved that
\begin{equation}
\mathcal{F}_1''(x,y,t)\leq CN^{3\alpha}|x-y|\, (|\log |x-y||+1). \label{f1}
\end{equation}

\bigskip
 In conclusion, recalling (\ref{F}), estimates (\ref{b1}) and (\ref{f1}) show that
\begin{equation}
\mathcal{F}_1(X^N(t), X^{N+1}(t),t)\leq  CN^{3\alpha}\delta^N(t)\left(1+|\log \delta^N(t)|\right). \label{f1'}
\end{equation}

\bigskip

\noindent
{\textbf{Estimate of the term}}  $\mathcal{F}_2$ 

\medskip

Putting $\bar{X}=X^{N+1}(t),$  we have:
\begin{equation}
\mathcal{F}_2(\bar X,t)\leq \mathcal{F}_2'(\bar X,t)+\mathcal{F}_2''(\bar X,t),
\label{I_2}
\end{equation}
where
\begin{equation}
\mathcal{F}_2'(\bar X, t)=\left|\int_{|\bar X - y|\leq 2\delta^N(t)} \frac{  \rho^N(y,t)-\rho^{N+1}(y,t)}{|\bar X-y|^2} \, dy \right|
\end{equation}
\begin{equation}
\mathcal{F}_2''(\bar X, t)=\left|\int_{2\delta^N(t)\leq |\bar X - y| } \frac{  \rho^N(y,t)-\rho^{N+1}(y,t)}{|\bar X-y|^2} \, dy \right|
\end{equation}
By estimate (\ref{rho_t}) we get
\begin{equation}
\mathcal{F}_2'(\bar X, t)\leq \int_{|\bar X - y|\leq 2\delta^N(t)} \frac{  \rho^N(y,t)+\rho^{N+1}(y,t)}{|\bar X-y|^2} \, dy \leq CN^{3\alpha}\delta^N(t).\label{F2'}
\end{equation}
Now we estimate the term $\mathcal{F}_2''.$ By the invariance of $f^N(t)$ along the characteristics, making a change of variables, we decompose the integral as follows:
\begin{equation}
\begin{split}
\mathcal{F}_2''(\bar X, t)=  \left| \int_{S^N(t)} d y \, dw\,\frac{f_0^N( y, w)}{|\bar X-  Y^N(t)|^2}\ -\int_{S^{N+1}(t)}dy \, dw\,\frac{f_0^{N+1}(y,w)}{|\bar X-Y^{N+1}(t)|^{2}} \right|
  \end{split}
\label{I_2''}
\end{equation}
where we put, for $i=N, N+1,$ $$(Y^i(t), W^i(t))= (X^i(y,w,t), V^i(y,w,t))$$ 
$$
S^{i}(t)=\{( y, w):2\delta^N(t)\leq|{\bar X}- Y^i(t)| \}.
$$
We notice that it is
\begin{equation} 
\mathcal{F}_2''(\bar X, t)\leq \mathcal{I}_1(\bar X, t)+\mathcal{I}_2(\bar X, t)+\mathcal{I}_3(\bar X, t)\label{l_i}
\end{equation}
where
\begin{equation}
\mathcal{I}_1(\bar X, t)= \int_{S^N(t)} d y  
 \int d w \left| \frac{1}{|\bar X -  Y^N(t)|^{2}}
-\frac{1}{|\bar X -  Y^{N+1}(t)|^{2}} \right| f_0^N( y,w),
\end{equation}
\begin{equation}
\mathcal{I}_2(\bar X, t)= \int_{S^{N+1}(t)} d y   
  \int d w \, \frac{ \left| f_0^N( y,w) - f_0^{N+1}( y, w) \right|}{|\bar X -  Y^{N+1}(t)|^{2}}  \, 
  \end{equation}
 \begin{equation}
  \mathcal{I}_3(\bar X, t)=
  \int_{S^N(t)\setminus S^{N+1}(t)}dy\int dw\,\frac{f_0^N(y,w)}{\left|\bar{X}-Y^{N+1}(t)\right|^2}.
  \end{equation}
We start by estimating $\mathcal{I}_1(\bar X, t).$ 
  By the Lagrange theorem 
 \begin{eqnarray}
&&\mathcal{I}_1(\bar X, t)\leq C |Y^{N}(t)- Y^{N+1}(t)| \int_{S^N(t)} dy\int dw \, \frac{ f_0^N(y,w)}{|\bar X - \xi^N(t) |^{3}}  \,
\end{eqnarray}
where $\xi^N(t)= \theta Y^{N}(t)+ (1-\theta) Y^{N+1}(t)$  for $\theta\in [0,1]$.
 By putting 
\begin{equation*}
\left(\bar{y}, \bar{w}\right)= \left(Y^N(t), W^N(t)\right)
\end{equation*}
and
$$
\bar S^N(t) =\{(\bar y, \bar w):  (y, w)\in S^N(t)  \},
$$
we get
\begin{equation}
\begin{split}
\mathcal{I}_1(\bar X, t)\leq & \,C \delta^N(t)\int_{ \bar S^N(t)}d\bar{y}\int d\bar{w} \, \frac{ f^N(\bar{y},\bar{w}, t)}{|\bar X - \xi^N(t) |^{3}}.\end{split}\label{bar}
\end{equation}

If $(y,w)\in S^N(t)$ then\begin{equation*}
 |\bar X -  \xi^N(t)| > |\bar X - \bar{y}| - |\bar{y} -Y^{N+1}(t)|\geq|\bar X - \bar{y}|-\delta^N(t)\geq \frac{|\bar X - \bar{y}|}{2}\label{>>}
\end{equation*}
which, by (\ref{bar}), implies
\begin{equation}
\begin{split}
\mathcal{I}_1(\bar X, t) \leq \ &C\delta^N(t)\int_{ \bar S^N(t)}d\bar{y}\int d\bar{w} \, \frac{ f^N(\bar{y},\bar{w}, t)}{|\bar X - \bar{y} |^{3}}\,=C\delta^N(t)\int_{ \bar S^N(t)}d\bar{y}\, \frac{ \rho^N(\bar{y}, t)}{|\bar X - \bar{y} |^{3}}
\end{split}.
\end{equation}
Now we consider the sets $$
A_1= \left\{(\bar y,\bar w):2 \delta^N(t)\leq\left|\bar{X}-\bar{y}\right|\leq  \frac{4}{\delta^N(t)}\right\}
$$
$$
A_2= \left\{(\bar y,\bar w): 1\leq\left|\bar{X}-\bar{y}\right|\leq 3R^N(t)\right\} $$
$$
 A_3=  \left\{(\bar y, \bar w):3R^N(t)\leq  |\bar{X}-\bar{y}|\right\}.
$$
Then it is $\bar S^N(t)\subset \bigcup_{i=1,2,3} A_i$ 
and
\begin{equation}
\mathcal{I}_1(\bar X, t) \leq C \delta^N(t)\sum_{i=1}^3\int_{  A_i}d\bar{y}\frac{ \rho^N(\bar{y}, t)}{|\bar X - \bar{y} |^{3}}.
\end{equation}
We estimate the integral over $A_1$ as we did in (\ref{D'}), the one over $A_2$ by means of Lemma \ref{L1} and the last one by means of Lemma \ref{L1'}, yielding
\begin{equation}
\begin{split}
\mathcal{I}_1(\bar X, t) \leq C\delta^N(t)\left[N^{3\alpha} |\log \delta^N(t)|+Q^N\left(3R^N(t),t\right)^{\frac35}+1\right].
\end{split}\label{i1}
\end{equation}
Equation (\ref{e'}) implies that
\begin{equation}
\mathcal{I}_1(\bar X, t) \leq C \delta^N(t)\left[N^{3\alpha}|\log \delta^N(t)|+Q^N\left(R^N(t),t\right)^{\frac35}+1\right]
\end{equation}
and by (\ref{Q^N}) in Corollary \ref{coro} we get
\begin{equation}
\begin{split}
\mathcal{I}_1(\bar X, t) \leq CN^{3\alpha} \delta^N(t)\left[|\log \delta^N(t)|+1\right]\end{split}.\label{i1'}
\end{equation}

We estimate now the term $\mathcal{I}_2(\bar X, t).$ By the definition of $f_0^i,$ for $i=N, N+1,$ and by (\ref{G}) it follows 
\begin{equation}
\begin{split}
&\mathcal{I}_2(\bar X, t)=
 \int_{S^{N+1}(t)} d y   
  \int_{} d w\,  \frac{  f_0^{N+1}( y, w) }{|\bar X -  Y^{N+1}(t)|^{2}} \chi\left(N\leq |w|\leq N+1\right) \leq\\& C_0\,e^{-\lambda N^2}\int_{S^{N+1}(t)} dy \int_{} dw\,\, \frac{g(|y|) }{|\bar X - Y^{N+1}(t)|^{2}} \chi\left(|w|\leq N+1\right).
\end{split} \label{AA} \end{equation}
We evaluate the integral over $S^{N+1}(t)$ by considering the sets  $$
B_1=   \left\{(y,w): |\bar{X}-Y^{N+1}(t)|\leq4R^{N+1}(t)\right\} $$

$$
 B_2=  \left\{(y,w):4R^{N+1}(t)\leq  |\bar{X}-Y^{N+1}(t)|\right\},
$$
so that 
\begin{equation}
S^{N+1}(t)\subset \bigcup_{i=1,2}B_i. \label{si}
\end{equation}
By putting 
\begin{equation*}
(\bar{y}, \bar{w})=\left(Y^{N+1}(t), W^{N+1}(t)\right),
\end{equation*}
 by (\ref{V^N}) it is $|\bar{w}|\leq CN,$ so that we have
\begin{equation}
\begin{split}
&\int_{B_{1}} dy \int_{} dw\, \frac{g(|y|) }{|\bar X - Y^{N+1}(t)|^{2}}\chi\left(|w|\leq N+1\right)\leq\\&C \int_{|\bar{w}|\leq C N} d\bar{w}\int_{ |\bar{X}-\bar{y}|\leq 4R^{N+1}(t)} d\bar{y}\, \frac{1 }{|\bar X - \bar{y}|^{2}}\leq CN^3R^{N+1}(t)\leq CN^4 .\label{g5}
\end{split}\end{equation}

For the integral over the set $B_2,$ we observe that, if $|\bar{X}-Y^{N+1}(t)|\geq 4R^{N+1}(t),$ then $|\bar{X}-y|\geq 3R^{N+1}(t)$ and  
\begin{equation*}
|\bar{X}-Y^{N+1}(t)|\geq |\bar{X}-y|-R^{N+1}(t)\geq \frac23|\bar{X}-y|.
\end{equation*} Hence
\begin{equation}
\begin{split}
&\int_{B_2 } dy \int_{} dw\, \frac{g(|y|) }{|\bar X - Y^{N+1}(t)|^{2}}\chi\left(|w|\leq N+1\right)\leq \\&\frac94\int_{ |\bar{X}-y|\geq 3R^{N+1}(t)} dy \int_{|w|\leq N+1} dw\, \frac{g(|y|) }{|\bar X -y|^{2}}\leq\\& CN^3\int_{|\bar{X}-y|\geq 3R^{N+1}(t)}  dy\, \frac{g(|y|) }{|\bar X -y|^{2}}.
\end{split}\end{equation}
The last integral is smaller than the analogous one already estimated in (\ref{g2}) and (\ref{s''}),  and is bounded by a constant. Hence we can conclude that
\begin{equation}
\int_{B_2 } dy \int_{} dw\, \frac{g(|y|) }{|\bar X - Y^{N+1}(t)|^{2}}\chi\left(|w|\leq N+1\right)\leq CN^3.\label{g6}
\end{equation}

Going back to (\ref{AA}), by (\ref{si}), (\ref{g5}) and (\ref{g6}) we have
\begin{equation}
\mathcal{I}_2(\bar X, t)\leq Ce^{-\lambda N^2}\left(N^4+N^3\right)\leq  Ce^{-\frac{\lambda}{2} N^2}.
\label{I}
\end{equation}
Finally, we estimate the term $\mathcal{I}_3(\bar X, t).$ If $(y,w)\in S^N(t),$ then
$$
\left|\bar{X}-Y^{N+1}(t)\right|\geq \left|\bar{X}-Y^{N}(t)\right| -\delta^N(t)\geq \delta^N(t),
$$
so that
 \begin{equation}
\begin{split}
 & \mathcal{I}_3(\bar X, t)=
  \int_{S^N(t)\setminus S^{N+1}(t)}dy\int dw\,\frac{f_0^N(y,w)}{\left|\bar{X}-Y^{N+1}(t)\right|^2}\leq\\&
\frac{1}{[\delta^N(t)]^2}  \int_{\left(S^{N+1}(t)\right)^c}dy\int dw\,f_0^N(y,w).
 \end{split} \end{equation}
 If $(y,w)\in \left(S^{N+1}(t)\right)^c,$ then $$\left|\bar{X}-Y^{N}(t)\right|\leq \left|\bar{X}-Y^{N+1}(t)\right|+ \delta^N(t)\leq 3\delta^N(t).$$ 
 Hence, by putting 
 \begin{equation*}
 (\bar{y}, \bar{w}) = \left(Y^N(t), W^N(t)\right)
 \end{equation*}
  we get
 \begin{equation}
\begin{split}
  \mathcal{I}_3(\bar X, t)\leq &
\,\frac{1}{[\delta^N(t)]^2}  \int_{\left|\bar{X}-\bar{y}\right|\leq 3\delta^N(t)} d\bar{y}\int d\bar{w}f^N(\bar{y}, \bar{w}, t)\leq\\&\frac{1}{[\delta^N(t)]^2}  \int_{ \left|\bar{X}-\bar{y}\right|\leq 3\delta^N(t)}d\bar{y}\,\rho(\bar{y},t)\leq \\& C\frac{N^{3\alpha}}{[\delta^N(t)]^2}  \int_{ \left|\bar{X}-\bar{y}\right|\leq 3\delta^N(t)}d\bar{y}\,\leq CN^{3\alpha}\delta^N(t).\end{split}\label{I3}
  \end{equation}
Going back to (\ref{l_i}), by (\ref{i1'}), (\ref{I}) and (\ref{I3}) we have
\begin{equation}
\mathcal{F}_2''(\bar X, t)\leq C\left[N^{3\alpha} \delta^N(t)\left(|\log \delta^N(t)|+1\right)+e^{-\frac{\lambda}{2} N^2}\right], \label{F2''}
\end{equation}
so that, by (\ref{I_2}), (\ref{F2'}) and (\ref{F2''}), we get
\begin{equation}
\mathcal{F}_2(\bar X, t)\leq C\left[N^{3\alpha} \delta^N(t)\left(|\log \delta^N(t)|+1\right)+e^{-\frac{\lambda}{2} N^2}\right].\label{f}
\end{equation}
Finally, (\ref{iter_}), (\ref{f1'}) and (\ref{f}) conclude the proof of the proposition.

\end{proof}

\bigskip

\section{Proofs of Theorems \ref{th_02} and \ref{th_03}}

\bigskip

\subsection{Proof of Theorem \ref{th_02}}

We put $V^i(t)=V^i(x,v,t)$ for $i=N, N+1$ and define
\begin{equation*}
\eta^N(t)= \sup_{(x,v)\in \mathbb{R}^3 \times B(0,N)} |V^N(t)-V^{N+1}(t)|
\end{equation*}
and
\begin{equation*}
\sigma^N(t)= \max\{\delta^N(t),\eta^N(t)\}.
\end{equation*}
Proposition \ref{p4} enables us to state the following result:
 \begin{proposition}
There exists a constant $C>1$ such that
\begin{equation}
\sup_{t\in [0,T]}\sigma^N(t)\leq C^{-CN}.
\end{equation}
\label{p5}
\end{proposition}

\begin{proof}  We start the proof for $\delta^N(t).$ To prove this proposition we need to iterate the estimate (\ref{estN}), inserting into the integral the same estimate for $\delta^N(t_2)$ with $t_2 \leq t.$ Unfortunately, being $E^N$ not Lipschitz continuous, estimate (\ref{estN}) is not linear in $\delta^N(t_2).$ However, it is easily seen that for any $r\in (0,1)$ and $a\in (0,1)$ the following inequality holds by convexity: 
$$
 r(|\log r|+1)\leq r|\log a|+a.
 $$
Hence, for $\delta^N(t_2)< 1,$ Proposition \ref{p4} gives us, for any $a<1,$
\begin{equation}
\delta^N(t) \leq C \int_0^t dt_1 \int_0^{t_1} dt_2\left[N^{3\alpha} \left( \delta^N(t_2)|\log a|+a\right)+e^{-\frac{\lambda}{2} N^2}\right].\label{ee}
\end{equation}
On the other side, were $\delta^N(t_2)\geq 1,$ Corollary \ref{coro} and Proposition \ref{propE} would provide a bound like 
\begin{equation}
\sup_{t\in[0,T]}\delta^N(t)\leq C N
\label{delta-N<CN}
\end{equation} 
and in that case estimate (\ref{estN}) would be simply
\begin{equation}
\delta^N(t)\leq C \int_0^t dt_1 \int_0^{t_1} dt_2\left[ N^{3\alpha}\log N \, \delta^N(t_2)+e^{-\frac{\lambda}{2} N^2}\right].
\end{equation}

In case $\delta^N(t_2)< 1,$ we choose $a =e^{-\frac{\lambda}{2} N^2}$ in (\ref{ee}), yielding
\begin{equation}
\begin{split}
\delta^N(t)\leq \,& C \int_0^t dt_1 \int_0^{t_1} dt_2\left[ N^{3\alpha+2}\delta^N(t_2)+N^{3\alpha}e^{-\frac{\lambda}{2} N^2}\right]\leq\\&  C \int_0^t dt_1 \int_0^{t_1} dt_2\left[ N^{3\alpha+2}\delta^N(t_2)+e^{-\frac{\lambda}{4} N^2}\right].
\end{split}\label{de}
\end{equation}
Note that the bound for $\delta^N(t_2)\geq 1$ is included in that for $\delta^N(t_2)< 1$.

\noindent Now we are in condition to iterate inequality (\ref{estN}), by inserting in the integral the same estimate for $\delta^N(t_2).$ We make $k$ iterations up to time $t_{2k+2}$ and at the last step we use the estimate (\ref{delta-N<CN}). Notice that at any step we get a double factorial at the denominator, due to the double time integration, obtaining
\begin{equation}
\delta^N(t)\leq CN\left[e^{-\frac{\lambda}{4} N^2}\sum_{i=1}^{k}\frac{C^iN^{(3\alpha+2)i}T^{2i}}{(2i)!}+\frac{C^kN^{(3\alpha+2)k}T^{2k}}{(2k)!}\right].
\end{equation}
The latter term is exponentially vanishing as $N\to \infty,$ provided $k$ has been chosen sufficiently large in function of $N,$ (for instance $k>N^{(3\alpha+2)}).$ The former term can be bounded by the series, that is
\begin{equation}
\sum_{i=1}^{k}\frac{C^iN^{(3\alpha+2)i}T^{2i}}{(2i)!}\leq e^{\sqrt{C}T N^{\left(\frac{3}{2}\alpha+1\right)}}.
\end{equation}
Hence
\begin{equation}
\delta^N(t)\leq CNe^{-\frac{\lambda}{4} N^2}e^{\sqrt{C}T N^{\left(\frac{3}{2}\alpha+1\right)}}+C^{-N}.\label{fin}
\end{equation}
We stress the fact that the double factorial, coming from the fact that the characteristics equation is a second order equation, allows us to halve  the exponent of the series. On the other hand, by Proposition \ref{propE}, it is $\frac{3}{2}\alpha+1< 2,$ so that we have
\begin{equation}
\delta^N(t)\leq C^{-CN}.
\end{equation}

The same arguments used in Proposition \ref{p4} allow us to prove, in analogy with (\ref{de}), that
 \begin{equation}
\begin{split}
\eta^N(t)&\leq C \int_0^t dt_1\left[ N^{3\alpha+2}\delta^N(t_1)+e^{-\frac{\lambda}{2} N^2}\right]\leq \\&CN^{3\alpha+2} \int_0^t dt_1\int_0^{t_1}dt_2\,\eta^N(t_2)+\int_0^tdt_1e^{-\frac{\lambda}{4} N^2},
\end{split}
\end{equation}
so that we can proceed as above in proving that
\begin{equation}
\eta^N(t)\leq C^{-CN},\label{fffin}
\end{equation}
which concludes the proof of the proposition.
\end{proof}

In the previous Proposition we have shown that 
$$
\max\left\{\sum_{N=1}^\infty \delta^N(t), \sum_{N=1}^\infty \eta^N(t)\right\}=C.
$$
 This implies that the sequences $\{X^N(x,v,t)\}$ and $ \{V^N(x,v,t)\}$ are Cauchy sequences and then, for any fixed $(x,v),$ they are bounded uniformly in $N.$
  Indeed, let us fix an initial datum $(x,v)$ and choose a positive integer $N_0$ such that $N_0=\hbox{intg}(|v|+C),$ where $\hbox{intg}(a)$ represents the integer part of the number $a$ and $C$ is assumed sufficiently large. Then, for any $N>N_0,$ by Proposition \ref{p5} and (\ref{V^N}) in Corollary \ref{coro}, we have 
\begin{equation}
\begin{split}
|V^N(t)-v|\leq\,& |V^{N_0}(t)-v|+\sum_{k=N_0}^{N-1}\eta^k(t)\leq\\& |V^{N_0}(t)-v|+C\leq |v| +CN_0 \end{split}
\end{equation}
which implies, by the choice of $N_0,$
\begin{equation}
|V^N(t)|\leq C(|v|+1)\label{N_0}
\end{equation}
and hence
\begin{equation}
|X^N(t)|\leq |x|+ C\left(|v|+1\right).\label{N_1}
\end{equation}
Since for any fixed $(x,v)$ the sequences $\{X^N(t)\},$ $\{V^N(t)\}$, and consequently $\{f^N(t)\}$, are  Cauchy sequences, they converge, as $N\to \infty,$ to some limit functions, which we call $X(x,v,t)$,  $V(x,v,t)$,  $f(x,v,t)$.

Properties (\ref{N_0}) and (\ref{N_1}) allow to show that the functions $f^N(t)$ enjoy properties (\ref{dect}) and (\ref{asst}), with constants $C$ and $\bar{\lambda}$ independent of $N.$ Indeed, (\ref{dect}) can be proved by (\ref{N_0}), observing that
\begin{equation}\begin{split}
f^N(X^N(t), V^N(t),t)=&f^N_0(x,v)\leq C_0e^{-\lambda |v|^2}g(x)\\& \leq  Ce^{-\bar{\lambda}|V^N(t)|^2}.\label{G2}
\end{split}\end{equation}
To prove that $\rho^N(t)$ satisfies (\ref{asst}) uniformly in $N,$ we fix $i\in{\mathbb{Z}}^3\setminus{\{0\}}$ and we decompose the velocity space into the sets $S_i$ and $S_i^c,$ where
$$
S_i=\left\{v\in{\mathbb{R}}^3: |v|\geq a_i= \sqrt{\frac{2(2+\epsilon)\ \log |i|}{\bar{\lambda}}}\right\}.
$$ 
Hence for any such $i$ it is
\begin{equation}
\begin{split}
\int_{|i-x|\leq1}\rho^N(x,t)\ dx=&\int_{|i-x|\leq1} dxdv f^N(x,v,t)=\\&\int_{|i-x|\leq1}dx\left[\int_{S_i}f^N(x,v,t)dv+\int_{S_{i}^c}f^N(x,v,t)dv\right].\label{g0}
\end{split}
\end{equation}
By (\ref{G2}) and the choice of $a_i$ we get  
 \begin{equation}
\begin{split}
\int_{|i-x|\leq1}dx&\int_{S_i}dvf^N(x,v,t)\leq C\int_{|i-x|\leq1}dx\int_{S_i}e^{-\bar{\lambda} |v|^2}dv\leq\\&e^{-\frac{\bar{\lambda}}{2} a_i^2}\int_{|i-x|\leq1}dx\int_{}e^{-\frac{\bar{\lambda}}{2} |v|^2}dv\leq  \frac{C}{|i|^{2+\epsilon}} .\end{split}\label{g1}
\end{equation}
On the other side, by (\ref{N_1}) we have
\begin{equation}
\begin{split}
&\int_{|i-x|\leq1}dx\int_{S_{i}^c}dvf^N(x,v,t)\leq\int_{|i-x|\leq Ca_i}\rho_0(x)\,dx\leq\\&
 C \sum_{\substack{ k\in \mathbb{Z}^3\\ |k|\leq Ca_i}}\int_{| i+k-x|\leq 1}g(x)\,dx\leq C\sum_{\substack{ k\in \mathbb{Z}^3\\ |k|\leq Ca_i}}\frac{1}{|i+k|^{2+\epsilon}}.
\end{split}\end{equation}
Since $|k|\leq Ca_i,$ by the definition of $a_i$ it follows that for large $|i|$ it is $|i+k|\geq \frac{|i|}{2},$ and then
\begin{equation*}
\sum_{\substack{ k\in \mathbb{Z}^3\\ |k|\leq Ca_i}}\frac{1}{|i+k|^{2+\epsilon}}\leq 2^{2+\epsilon}\sum_{\substack{ k\in \mathbb{Z}^3\\ |k|\leq Ca_i}}\frac{1}{|i|^{2+\epsilon}}\leq C \frac{(\log|i|)^{\frac32}}{|i|^{2+\epsilon}},
\end{equation*}
so that we get
\begin{equation}
\int_{|i-x|\leq1}dx\int_{S_{i}^c}dvf^N(x,v,t) \leq  C\, \frac{(\log|i|)^{\frac32}}{|i|^{2+\epsilon}}
\end{equation}
which, together with (\ref{g0}) and (\ref{g1}), proves (\ref{asst}).
This implies that also the limit function $f(x,v,t)$ satisfies the same properties (\ref{dect}) and (\ref{asst}). 

\noindent It remains to prove that the couple $\left(X(x,v,t), V(x,v,t)\right)$ satisfies the right equation, that is we have to show that $|E^N(x,t)- E(x,t)|\to 0$ as $N\to \infty$,  with $E(x,t)$ defined in (\ref{Eq}).
To prove this, we note first that from estimate (\ref{G2}) it follows
\begin{equation}
\| \rho^N(t)\|_{L^\infty}\leq C.\label{lim}
\end{equation}
The term $|E^N(x,t)- E(x,t)|$  can be handled as the term $\mathcal{F}_2$ in the previous section, by using the bound (\ref{lim}), yielding
\begin{equation}
|E^N(x,t)- E(x,t)|\leq C\left[ \delta^N(t)\left(1+|\log \delta^N(t)|\right)+e^{-\frac{\lambda}{2} N^2}\right]
\end{equation}
which is infinitesimal as $N\to \infty,$ by Proposition \ref{p5}.
Uniqueness could be proved along the same lines and this
concludes the proof of the Theorem \ref{th_02}. 

\bigskip

\subsection{Proof of Theorem \ref{th_03}}

To prove Theorem \ref{th_03} we only need to prove that (\ref{dec}) holds for any $t\in [0,T],$ under the hypotheses (\ref{Ga}) and (\ref{asp}), since existence and uniqueness of the solution has already been proved in the preceding subsection. To this purpose, we write again:
\begin{equation}\begin{split}
f^N(X^N(t), V^N(t),t)=f^N_0(x,v)\leq f_0(x,v)\leq C_0e^{-\lambda |v|^2}g(|x|).\label{G2'}
\end{split}\end{equation}
Calling $C^*$ the constant appearing in (\ref{N_1}), we see that 
\begin{equation}
|x|\geq \left|X^N(t)\right|-C^*(|v|+1).\label{x}
\end{equation}
Now we consider the two possible cases:
\begin{equation}
C^*(|v|+1)\leq \frac12 \left|X^N(t)\right|
\end{equation}
and
\begin{equation}
C^*(|v|+1)\geq \frac12 \left|X^N(t)\right|.
\end{equation}
In the first case, by (\ref{x}) it is $|x|\geq\frac12 \left|X^N(t)\right|$ and consequently, 
by the properties of $g$ defined in (\ref{asp})
 we have
\begin{equation}
g(|x|)\leq g\left(\frac12\left|X^N(t)\right|\right)\leq Cg\left(\left|X^N(t)\right|\right).
\end{equation}
From here, going back to (\ref{G2'}) and recalling (\ref{N_0}), there exists $\lambda'>0$ such that
\begin{equation}\begin{split}
f^N(X^N(t), V^N(t),t)\leq  Ce^{-\lambda' \left|V^N(t)\right|^2}g\left(\left|X^N(t)\right|\right).\label{G'}
\end{split}\end{equation}
Let us now consider the second case. By using again (\ref{N_0}) in (\ref{G2'}), we have
\begin{equation*}\begin{split}
f^N(X^N(t), V^N(t),t)\leq&\, C_0e^{-\frac{\lambda}{2} |v|^2}e^{-\lambda'|V^N(t)|^2}g(|x|)\leq\\& Ce^{-\frac{\lambda}{2} |v|^2}e^{-\lambda'|V^N(t)|^2}.\label{}
\end{split}\end{equation*}
Notice that in the case at hand it is $|v|\geq C\left|X^N(t)\right|$ and hence
\begin{equation*}\begin{split}
f^N(X^N(t), V^N(t),t)\leq Ce^{- C\left|X^N(t)\right|^2}e^{-\lambda'|V^N(t)|^2}.\label{}
\end{split}\end{equation*}
Since  for any positive $r$  it results   $e^{-r^2}\leq \frac{C}{r^{2+\epsilon}}$,  we have also in this case
\begin{equation}\begin{split}
f^N(X^N(t), V^N(t),t)\leq  Ce^{-\lambda' \left|V^N(t)\right|^2}g\left(\left|X^N(t)\right|\right)\label{g'}
\end{split}\end{equation}
provided the constant in (\ref{asp}) is sufficiently large.

Estimates (\ref{G'}) and (\ref{g'}) conclude the proof of Theorem  \ref{th_03}.

\section{Proof of Proposition \ref{propE}}
\subsection{Proof of Proposition \ref{propE} for $\epsilon>13/19$}

This proposition will be proved in complete analogy to what done in \cite{R3}, with the only
effort to lower the exponent of the estimate, from $1$ to $2/3$.
We repeat here the proof, choosing accurately  some parameters, 
and from now on we simplify the notation skipping the index $N$. We also put
\begin{equation}
 P={\mathcal{V}}^N(T)
 \label{P}
 \end{equation}
 \begin{equation}
 Q=\sup_{t\in[0,T]}Q^N(R^N(t),t).
 \end{equation}
We define
\begin{equation}
\beta=1-\epsilon,
\end{equation}
being $\epsilon>1/15$ the parameter in (\ref{ass}), consequently it is $\beta<14/15$.
We divide the proof in a first part in which $\beta<6/19$ (that is, $\epsilon>13/19$), which
does not require an iterative method on a suitable time average, and after we improve the result
up to $\beta<14/15$ (that is, $\epsilon>1/15$) by an iterative method.
 
We fix a time interval
\begin{equation}
\Delta=\frac{1}{4C_2P^{\frac43-\gamma}Q^{\frac13}}  \label{Delta}
\end{equation}
where $C_2$ is the constant in (\ref{C1}) and $\gamma$ is any number satisfying
\begin{equation}
\frac{4}{3}\beta <\gamma<\frac{2-\beta}{4}
\label{gg}
\end{equation} 
(the reason for this range of the parameter $\gamma$ will be clear in the following;
we remark that when $\beta\to \frac{6}{19}$ the lower bound tends to the upper bound). 
We note that such choice  (\ref{Delta}) assures that  $\Delta \ll T$ (taking the constant $\widetilde C$
in (\ref{mv}) suitably large); indeed  it is $\gamma<\frac43$ 
and $P$ is defined to be greater than  $\widetilde C$. 

Let us consider two solutions to the $N$-truncated system, $\left(X(t),V(t)\right)$ and $\left(Y(t),W(t)\right).$ The following results, whose proofs are given in the Appendix, hold.
\begin{lemma}\label{lemv3}
 Let $t'\in [0,T].$
$$
\hbox{If} \qquad |V(t')-W(t')|\leq P^{\gamma} $$
then 
\begin{equation}
\sup_{t\in [t',t'+\Delta]}|V(t)-W(t)|\leq 2P^{\gamma}.
\end{equation}

\medskip

$$
\hbox{If} \qquad|V(t')-W(t')|\geq P^{\gamma} 
$$
then 
\begin{equation}
\inf_{t\in [t',t'+\Delta]}|V(t)-W(t)|\geq \frac12P^{\gamma}
\end{equation}
\end{lemma}

\medskip

\begin{lemma}\label{lemrect}
Let $t'\in [0,T]$  and assume that $|V(t')-W(t')|\geq P^\gamma.$
Then  there  exists $t_0\in [t',t'+\Delta]$ such that for any $t\in [t',t'+\Delta]:$
$$
|X(t)-Y(t)|\geq \frac{P^\gamma}{4}|t-t_0|.
$$
\end{lemma}

\bigskip

\bigskip

\noindent Let us divide the interval $[0,T]$ into $n$ sub-intervals $[t_i,t_{i+1}],$  $i=0,1,...,n-1$,
 with $t_0=0,$ $t_n=T$ and  $\frac12 \Delta \leq t_{i+1}-t_i \leq \Delta.$ Hence it is:
\begin{equation}
\int_0^t | E(X(s),s)| \, ds=\sum_{i=0}^{n-1}\int_{t_i}^{t_{i+1}}| E(X(s),s)| \, ds.
\end{equation}
Fixing the index $i$, we consider the time evolution of the system over the time interval 
$[t_i, t_{i+1}].$ 
We have
\begin{equation}
|E(X(t),t)|\leq \int dydw \ \frac{ f(y, w,t)}{|X(t)- y|^2}
\label{Ei}
\end{equation}
and we decompose the phase space $(y,w)\in \mathbb{R}^3\times \mathbb{R}^3$ in the following way, being $\gamma $ the positive parameter introduced in (\ref{Delta}), (\ref{gg}):
\begin{equation}
A=\{{(y,w)}: |X(t_i)-{y}|\leq 4R(T)\}
\label{T1}
\end{equation}
\begin{equation}
B_1= \{ {(y,w)}: |V(t_i)-{w}|\leq P^\gamma \}\label{S1}
\end{equation}
\begin{equation}
B_2=\{ {(y,w)}: \  |V(t_i)-{w}|> P^\gamma \}.\label{S3}
\end{equation}

Hence we have, for any $t\in [t_i,t_{i+1}]$,
\begin{equation}
|E(X(t),t)|\leq \sum_{j=1}^3{\mathcal{I}}_j(X(t))\label{sum}
\end{equation}
where for $ j=1,2$
\begin{equation*}
{\mathcal{I}}_j(X(t))=  \int _{A\cap B_j}dydw \frac{ f(y, w, t)}{|X(t)- y|^2} \end{equation*}
and 
\begin{equation*}
{\mathcal{I}}_3(X(t))= \int _{A^c}dydw \frac{ f(y, w, t)}{|X(t)-y|^2}
\end{equation*}
being
 \begin{equation*}
A^c=\{{(y,w)}: |X(t_i)-{y}| > 4R(T)\}.
\end{equation*}
Let us start by the first integral. Proceeding as in the proof of (\ref{E1}), with ${\mathcal{I}}_1(X(t))$ in place of ${\mathcal{J}}_1(X(t),t)+{\mathcal{J}}_2(X(t),t),$ and enlarging the region of integration to
$A'= \{{(y,w)}: |X(t)-{y}|\leq 5R(T)\}$,  we get
\begin{equation}
{\mathcal{I}}_1(X(t))\leq  C P^{\frac43\gamma}Q^{\frac13}. \label{I01}
\end{equation}

\bigskip

Now we analyse  the term ${\mathcal{I}}_2.$ 
We introduce the sets $C_{k}(t)$ and  $D_{k}(t)$,  $t\in [t_i,t_{i+1}]$,  with ${k=0,1,2,...,m}$,
which constitute a finer decomposition of the set $B_2$ and are defined in the following way:
\begin{equation}
\begin{split}
C_{k}(t)=\big\{ (y, w):  \  \alpha_{k+1}< |V(t_i)- w|\leq \alpha_k, \  |X(t)- y|\leq l_{k}\big\}
\end{split}\label{Ak}
\end{equation}
\begin{equation}
\begin{split} 
D_{k}(t)=\big\{ (y, w): \   \alpha_{k+1}< |V(t_i)-w|\leq \alpha_k, \  |X(t)-y|> l_{k}\big\}
\end{split}\label{Bk}
\end{equation}
and we choose the parameters $\alpha_k$ and $l_k$ as:
\begin{equation}
\alpha_k=\frac{P}{2^k} \quad \quad l_{k}=\frac{2^{3k}Q^{\frac13}}{P^{\frac43+\eta}}
\label{al}
\end{equation}
with $\eta$  any number satisfying
\begin{equation}
\frac{3+\beta}{3}<\eta< 1+\gamma-\beta
\label{eta}
\end{equation} 
(the reason for this range of the parameter $\eta$ will be clear in the following;
we remark that, due to the choice of $\gamma$ (\ref{gg}), this interval is well defined, 
and when $\beta\to \frac{6}{19}$ the lower bound tends to the upper bound).
Consequently we put
\begin{equation}
{\cal{I}}_2(X(t))\leq \sum_{k=0}^m \left( {\cal{I}}_2'(k)+{\cal{I}}_2''(k) \right) \label{23}
\end{equation}
being
\begin{equation}
{\cal{I}}_2'(k)=\int_{A\cap C_{k}(t)} \frac{f(y, w, t)}{|X(t)- y|^2} \,dydw
\end{equation}
and
\begin{equation}
{\cal{I}}_2''(k)=\int_{A\cap D_{k}(t)} \frac{f(y, w, t)}{|X(t)-y|^2} \, dydw.
\end{equation}

By the choice of the parameters $\alpha_k$ and $l_{k}$ made in (\ref{al}), we have:

\begin{equation}
\begin{split}{\cal{I}}_2' (k)\leq
C \, l_{k}\int_{\alpha_{k+1}< |V(t_i)- w|\leq \alpha_k} \, dw  \leq Cl_{k}\alpha^3_{k}\leq C  P^{\frac53-\eta}Q^{\frac13}. \end{split}
\label{int3}
\end{equation}
Hence
\begin{equation}
 \sum_{k=0}^m{\cal{I}}_2'(k)\leq C  P^{\frac53-\eta}Q^{\frac13} \log P,
\label{i3}\end{equation}
since, considering we are in the set $B_2,$ it is 
\begin{equation}
m\leq  (1-\gamma)\log_2 P.\label{par}
\end{equation}

Now we pass to ${\mathcal{I}}_2''(k)$, for which we need to make the time integration over the interval $[t_i,t_{i+1}]$.
We set briefly
$$
(Y(t),W(t)):=(Y(t,t_i,\bar y,\bar w),W(t,t_i,\bar y,\bar w))
$$
being
$$
 ( Y(t_i), W(t_i)) = (\bar y,\bar w),
$$
hence we have
\begin{equation}
\begin{split}
\int_{t_i}^{t_{i+1}}& {\mathcal{I}}_2''(k)\ dt\leq\int_{t_i}^{t_{i+1}}dt\int_{A\cap D_{k}(t)} \ \frac{ f(y,w,t)}{|X(t)-y|^2} 
 \, dydw =
\\&\int_{t_i}^{t_{i+1}}dt\int_{\widetilde{A\cap D_{k}(t)}} \ \frac{ f(\bar y,\bar w,t_i)}{|X(t)-Y(t)|^2} 
\, \chi(|X(t)-Y(t)|>l_k) \, d\bar y d\bar w
\end{split}
\label{ik}
\end{equation}
under the change of variables $(y,w)=(Y,W)(t,t_i,\bar y, \bar w)$, being $\widetilde{A\cap D_{k}(t)}$
the backward in time evolved set of $A\cap D_{k}(t)$ at time $t_i$,
recalling the invariance of $f$ along the characteristics and the conservation of the measure of the phase space.

\noindent We note that we enlarge the set $\widetilde{A\cap D_{k}(t)}$ by integrating over the set
$A'\cap D_k'$, where
$$
A' = \{{(\bar y,\bar w)}: |X(t_i)-{\bar y}|\leq 5R(T)\},
$$
$$
D_k' = \{{(\bar y,\bar w)}:  \frac{\alpha_{k+1}}{2} < |V(t_i)-\bar w| \leq 2\alpha_k   \},
$$
since by the definition of the maximum displacement $R(T)$ and by Lemma \ref{lemv3} (slightly adapted),
particles belonging to $A$ and $D_k(t)$ at time $t$ come necessarily from $A'$ and $D_k'$ at time $t_i$.
This observation allows us to change the order of integration in (\ref{ik}), arriving at 
\begin{equation}
\int_{t_i}^{t_{i+1}} {\mathcal{I}}_2''(k)\ dt\leq\int_{A'\cap D_k'}f(\bar y, \bar w, t_i) 
\left( \int_{t_i}^{t_{i+1}} \frac{\chi(|X(t)-Y(t)|>l_k)}{|X(t)-Y(t)|^2} \, dt \right) d\bar y d\bar w.
\label{4.26}
\end{equation}

\noindent By rephrasing Lemma \ref{lemv3} in this case it is easily seen that it holds:
\begin{equation}
\forall \ (\bar y,\bar w)\in D_{k}'\qquad  |V(t)-W(t)|\geq \frac{\alpha_{k+1}}{4}\label{lemperp1}
\end{equation} 
and consequently, by Lemma \ref{lemrect}, there exists $t_0\in [t_i,t_{i+1}]$ such that, for any $t\in [t_{i},t_{i+1}]:$
\begin{equation}
|X(t)-Y(t)|\geq \frac{\alpha_{k+1}}{16}|t-t_0|. \label{parr}
\end{equation}
Hence, putting $a =  l_{k} /\alpha_k,$  we have:
\begin{equation}
\begin{split}
&\int_{t_i}^{t_{i+1}} \frac{\chi(|X(t)-Y(t)|>l_k)}{|X(t)-Y(t)|^2} \, dt 
\leq  \,\,   \\
&\int_{\{ t: |t-t_0|\leq a \}} \frac{\chi(|X(t)-Y(t)|>l_k)}{|X(t)-Y(t)|^2} \, dt \,+ \\
&\int_{\{ t: |t-t_0| > a \}} \frac{\chi(|X(t)-Y(t)|>l_k)}{|X(t)-Y(t)|^2} \, dt   \leq  \\
&\frac{1}{l_k^2}\int_{\{ t: |t-t_0|\leq a \}} \, dt
+  \frac{C}{\alpha_{k+1}^2}\int_{\{ t: |t-t_0| > a \}} \frac{1}{|t-t_0|^2} \, dt \leq \\
& \quad \frac{2 a}{l_{k}^2} + \frac{C}{\alpha_{k+1}^2} \int_a^{+\infty} \frac{1}{t^2} \, dt
= \frac{C}{\alpha_{k} l_k}.
\end{split}
\label{eq1}\end{equation}
Moreover, defining $\widetilde y=\bar y$,  $\widetilde w = \bar w -V(t_i)$, and 
$$D_k'' = \{{(\widetilde y,\widetilde w)}:  \frac{\alpha_{k+1}}{2} < |\widetilde w| \leq 2\alpha_k   \},
$$
\begin{equation}
\begin{split}
\int_{A'\cap D_{k}''} f(\widetilde{y}, \widetilde{w}, t_i)\ d\widetilde{y}d\widetilde{w}&\leq \frac{C}{\alpha_k^2}\int_{A'\cap D_{k}''} \widetilde{w}^2 f(\widetilde{y}, \widetilde{w}, t_i)\ d\widetilde{y}d\widetilde{w} \\& \end{split}
\label{eq2}
\end{equation}
so that by (\ref{ik}), (\ref{4.26}), (\ref{eq1}) and (\ref{eq2}), taking into account (\ref{al}), we get:
\begin{equation}
\begin{split}
\sum_{k=0}^m\int_{t_i}^{t_{i+1}} {\mathcal{I}}_2''(k)&\leq\sum_{k=0}^m\frac{C}{\alpha_{k}^3 l_k}\int_{A'\cap D_{k}''} \widetilde{w}^2 f(\widetilde{y}, \widetilde{w}, t_i)\ d\widetilde{y}d\widetilde{w}\leq \\&\frac{C}{P^{\frac53-\eta}Q^{\frac13}}\sum_{k=0}^m\int_{A'\cap D_{k}''} \widetilde{w}^2 f(\widetilde{y}, \widetilde{w}, t_i)\ d\widetilde{y}d\widetilde{w}.\end{split}
\end{equation}
Now notice that:
\begin{equation}
\begin{split}
\sum_{k=0}^m&\int_{A'\cap D_{k}''} \widetilde{w}^2 f(\widetilde{y}, \widetilde{w}, t_i)\, d\widetilde{y}d\widetilde{w}\leq \int_{A'} \widetilde{w}^2 f(\widetilde{y}, \widetilde{w}, t_i)\ d\widetilde{y}d\widetilde{w}\leq\\& W(X(t_i),5R(T),t_i)\leq CQ(R(T),t_i)\leq CQ,
\end{split}
\label{i4}
\end{equation}
as it follows from  (\ref{R'/R}), and this implies:
\begin{equation}
 \sum_{k=0}^m\int_{t_i}^{t_{i+1}} {\mathcal{I}}_2''(k)\ dt\leq  \frac{CQ^{\frac23}}{ P^{\frac53-\eta}}. \label{i3'}
\end{equation}

From eqns. (\ref{I01}), (\ref{23}), (\ref{i3}) and (\ref{i3'}) it follows:
\begin{equation}
\begin{split}
\sum_{j=1}^2\int_{t_i}^{t_{i+1}}\ {\mathcal{I}}_j(X(t)) \ dt\leq  C P^{\frac43\gamma}Q^{\frac13}\Delta+ C  P^{{\frac53-\eta}}Q^{\frac13} \log P\Delta+\frac{CQ^{\frac23}}{ P^{{\frac53-\eta}}}.\end{split}\label{p}
\end{equation}
For the last term ${\mathcal{I}}_3(X(t)),$ we observe that if $|X(t_i)-y|\geq 4R(T),$ then $|X(t)-y|\geq |X(t_i)-y|-R(T)\geq 3R(T).$ Hence
$$
{\mathcal{I}}_3(X(t))=\int _{A^c}dydw \frac{ f(y, w,t)}{|X(t)-y|^2}\leq
 \int _{|X(t)-y|\geq3R(T)}dydw \frac{ f(y, w, t)}{|X(t)-y|^2}.
$$
Along the same lines of the estimate of ${\mathcal{J}}_3(x,t)$ in Proposition \ref{prop2}  we can proceed with the estimate of ${\mathcal{I}}_3(X(t))$, obtaining: 
\begin{equation}
{\mathcal{I}}_3(X(t))\leq C.
\label{i6}
\end{equation}
By (\ref{sum}), (\ref{p}) and (\ref{i6}), recalling the definition (\ref{Delta}) of the time interval $\Delta,$ we have:
\begin{equation}
\int_{t_i}^{t_{i+1}}|E(X(t),t)|dt\leq C\Delta \left(P^{\frac43\gamma}Q^{\frac13}+ C  P^{{\frac53-\eta}}Q^{\frac13} \log P+\frac{CQ}{ P^{{\frac13-\eta+\gamma}}}\right).\label{av_01}
\end{equation}

\noindent Proposition \ref{propo},  Lemma  \ref{luc},  and the definition of $R(t)$ imply that
\begin{equation}
Q\leq C P^{\beta} \label{qpiccolo}
\end{equation}
so that from (\ref{av_01}) it follows:
\begin{equation}
\int_{t_i}^{t_{i+1}}\ |E(X(t),t)| \ dt\leq  C\Delta \left(P^{\frac43\gamma+\frac{\beta}{3}}+  P^{\frac53+\frac{\beta}{3}-\eta} \log P+ P^{\beta-\frac13-\gamma+\eta}\right).\label{i00}
\end{equation}

Hence, since $n\Delta=T,$ by (\ref{i00}) we get:
\begin{equation}
\begin{split}
 \int_0^t |E(X(s),s)| &\, ds\leq \sum_{i=0}^{n-1}\int_{t_i}^{t_{i+1}}\ |E(X(t),t)|\ dt\leq\\&
 C T\left(P^{\frac43\gamma+\frac{\beta}{3}}+  P^{\frac53+\frac{\beta}{3}-\eta} \log P+P^{\beta-\frac{1}{3}-\gamma+\eta}\right).\label{fsez4}
\end{split}
\end{equation}

By the assumption on $\beta<\frac{6}{19}$ and the choices of the parameters made in (\ref{gg}) and (\ref{eta})  we finally have:
 \begin{equation}
 \int_0^t |E(X(s),s)| \, ds\leq C P^{\alpha},    \label{EP}
\end{equation}
and the thesis is proved.

\subsection{Proof of Proposition \ref{propE}  for $\frac{6}{19}\leq\beta<\frac{14}{15}$}

In the case $\frac{6}{19}\leq\beta<\frac{14}{15}$ it 
is not  possible to choose  the parameters $\gamma$ and $\eta$ in (\ref{fsez4}) in order to prove (\ref{EP}). 
What we do then is performing further iterations of the estimate of the time average of the electric field (\ref{av_01}), by using at each step a time interval
larger than the previous one; in such a way we obtain at any step an improved estimate, up to the achievement of the searched bound. 
More precisely, for a positive integer $\ell$ we define
\begin{equation}
\Delta_\ell=\Delta_{\ell -1} {\mathcal{G}}=\Delta_{\ell -2}
{\mathcal{G}}^2=...=\Delta_{1}{\mathcal{G}}^{\ell-1}
\end{equation}
where  $\Delta_{1}:=\Delta$  (defined in (\ref{Delta})), 
\begin{equation}
\mathcal{G}=\textnormal{intg}\left(P^{\delta}\right),
\label{G_factor}
\end{equation}
being $\textnormal{intg}(a)$ the integer part of $a,$ and  $\delta$ is any number satisfying
\begin{equation}
0<\delta< \frac{7}{6}-\frac54\beta.\label{delta4}
\end{equation}
The reason for such bounds relies in the proof of Lemma \ref{lemv3l} in the Appendix.
The integer part in (\ref{G_factor}) is taken in order to use  well known
properties of the average over intervals which increase by an integer factor (see Remark \ref{rem_av} in the Appendix).

We succeed in the control of the time average of the electric field over a suitable time interval,
as stated in the following proposition:
\begin{proposition}
There exists a positive constant $\bar\Delta$ such that
\begin{equation}
\langle E \rangle_{\bar\Delta} := \frac{1}{\bar\Delta} \int_t^{t+\bar\Delta} |E(X(s), s)| \, ds
\leq C P^{\alpha} \qquad \forall \,\, \alpha \in \left(\frac{5-\epsilon}{9}, \frac23 \right)
\label{iteraz}
\end{equation}
for any $t\in[0, T]$ such that  $t\leq T-\bar\Delta$.
\label{5}
\end{proposition}
\begin{proof}
We claim that the  following estimate holds for any positive integer $\ell$
(its proof is given in the following subsection):
\begin{equation}
\langle E \rangle_{\Delta_\ell} \leq C \left[  P^{\frac43\gamma}Q^{\frac13} +
P^{\frac53-\eta} Q^\frac13 \log P + \frac{Q}{P^{\frac13-\eta+\gamma +(\ell-1)\delta}}  \right].
\label{av_ell}
\end{equation}
Notice that here we get an improved bound by a factor $P^{-(\ell-1)\delta}$ with respect to (\ref{av_01}). From this, by (\ref{qpiccolo}) it follows
\begin{equation}
\langle E \rangle_{\Delta_\ell} \leq C \left[P^{\frac43\gamma+\frac\beta3} +
P^{\frac53+\frac\beta3-\eta}  \log P + P^{\beta-\frac13+\eta-\gamma-(\ell-1)\delta} \right].
\label{passo_ell}
\end{equation}
A differente choice of the parameters $\gamma$ and $\eta$ (with respect (\ref{gg})  and  (\ref{eta})) is necessary for the case $\frac{6}{19}\leq\beta<\frac{14}{15}$, in order to prove also Lemmas  \ref{lemv3l}  and  \ref{lemrectl} (later on):
\begin{equation}
\max\left\{ 0, \beta-\frac23+\delta \right\} < \gamma < \frac{2-\beta}{4},
\label{gamma_iter}
\end{equation}
and
\begin{equation}
\frac{3+\beta}{3}<\eta<\frac53+\gamma-\frac23\beta-\delta.
\label{eta_iter}
\end{equation}
 Defining $\bar\ell$ as the smallest integer such that
\begin{equation}
\beta-\frac13+\eta-\gamma-(\bar\ell-1)\delta < \frac23,
\label{elle_bar}
\end{equation}
estimate  (\ref{passo_ell})  implies  (\ref{iteraz})  with $\bar\Delta = \Delta_{\bar\ell}$.
\end{proof}

\medskip

\noindent It can be easily seen that 
$$
C\frac{P^{\beta+\eta-\frac73}}{Q^{\frac13}}<\bar\Delta<C\frac{P^{\beta+\eta+\delta-\frac73}}{Q^{\frac13}}
$$
and by the choice of the parameters both the exponents of $P$ are negative, in order to be $\bar\Delta \ll T$
(taking the constant $\widetilde C$
in (\ref{mv}) suitably large).

\subsection{Proof of (\ref{av_ell}).}
By an inductive method we give now the proof of  (\ref{av_ell}).

\noindent We premise the following Lemmas, direct generalizations of Lemmas \ref{lemv3} and \ref{lemrect}, whose statements hold 
also at the $\ell$-th iterative step. Their proofs are given in \cite{R3}, but for completeness  we repeat them in the Appendix:

\begin{lemma}\label{lemv3l}
 Let $t\in [0,T]$ such that $t+\Delta_\ell\in [0,T]$.
$$
\hbox{If} \qquad |V(t)-W(t)|\leq P^{\gamma} $$
then 
\begin{equation}
\sup_{s\in [t, t+\Delta_\ell]}|V(s)-W(s)|\leq 2P^{\gamma}.
\end{equation}
$$
\hbox{If} \qquad|V(t)-W(t)|\geq P^{\gamma}
$$
then 
\begin{equation}
\inf_{s\in [t, t+\Delta_\ell]}|V(s)-W(s)|\geq \frac12 P^{\gamma}. \label{L2} 
\end{equation}
\end{lemma}
\begin{lemma}\label{lemrectl}
 Let $t\in [0,T]$ such that  $t+\Delta_\ell\in [0,T]$       
  and assume that $|V(t)-W(t)|\geq  P^{\gamma}$. 
Then there  exists $t_0\in [t, t+\Delta_\ell]$ such that for any $s\in [t, t+\Delta_\ell] $ it holds:
$$
|X(s)-Y(s)|\geq \frac{P^{\gamma}}{4}|s-t_0|.
$$
\end{lemma} 
Formula (\ref{av_ell}) is already proved in the case $\ell=1$  (it is in fact formula (\ref{av_01})).
We show now that, assuming true (\ref{av_ell}) for $\ell-1$, then it holds also for $\ell$.
  We note that (see Remark \ref{rem_av} in the Appendix)  if estimate (\ref{av_ell}) holds for $\langle E \rangle_{\Delta_{\ell-1}}$, then the same estimate holds for $\langle E \rangle_{\Delta_{\ell}},$   since both the intervals $\Delta_{\ell}$ and the bound (\ref{av_ell}) are uniform in time. 
From what done before,   the only term in (\ref{av_ell})  which is affected
by the time average is the last one (see (\ref{i3'})). Then, by using the estimate on the time average 
$\langle E \rangle_{\Delta_{\ell-1}}$ on the larger time interval $\Delta_\ell$  we get for this term:
\begin{equation}
\frac{Q}{P^{\frac13+\gamma-\eta +(\ell-2)\delta}} \Delta_{\ell} \frac{\Delta_{\ell-1}}{\Delta_\ell}
\leq \frac{Q}{P^{\frac13-\eta+\gamma +(\ell-1)\delta}} \Delta_{\ell},
\end{equation}
which proves (\ref{av_ell}).
\qed

\medskip

The proof of Proposition \ref{propE} follows immediately, as in the passage from (\ref{fsez4})
to (\ref{EP}).

\bigskip

\bigskip

\bigskip

\subsection*{Acknowledgments}
Work performed under the auspices of 
GNFM-INDAM and the Italian Ministry of the University (MIUR).  

\bigskip

\bigskip

\section*{Appendix}

\appendix

\setcounter{equation}{0}

\def\theequation{A.\arabic{equation}}

\medskip

\subsection*{Proof of  Lemma \ref{lemv3}}

We have
by (\ref{C1}),  for any $t\in[t', t'+\Delta]\subset [0,T]$,
\begin{equation*}
\begin{split}
|V(t) - W(t)| &\leq |V(t')-W(t')| +\\& \int_{t'}^{t'+\Delta} \Big[ |E(X(s),s)| + |E(Y(s),s)| \Big] ds 
\leq \\&P^\gamma + 2 C_2 P^\frac43 Q^\frac13 \Delta \leq 2 P^\gamma.
\end{split}
\end{equation*}

Analogously we prove the second statement:
\begin{equation*}
\begin{split}
|V(t)-W(t)|&\geq |V(t')-W(t')|-\\&\int_{t'}^{t'+\Delta}
\Big[ |E(X(s),s)| +|E(Y(s),s)| \Big] ds
\geq\\& P^\gamma - 2 C_2 P^{\frac43} Q^\frac13  \Delta
\geq \frac12 P^\gamma.
\end{split}
\end{equation*}

\subsection*{Proof of  Lemma \ref{lemrect}}

Let $t_0\in [t', t'+\Delta]\subset [0,T]$ be the time at which $|X(s)-Y(s)|$ has the minimum value,
for $s\in[t', t'+\Delta]$.   We put $\Gamma(s)=X(s)-Y(s)$.  
Moreover we define the function
$$
\bar{\Gamma}(s)=\Gamma(t_0)+ \dot{\Gamma}(t_0)(s-t_0).
$$
It results
$$
\ddot{\Gamma}(s)-\ddot{\bar{\Gamma}}(s)=E(X(s),s)-E(Y(s),s)
$$
$$
\Gamma(t_0)=\bar{\Gamma}(t_0), \quad  \dot{\Gamma}(t_0)=\dot{\bar{\Gamma}}(t_0)
$$
from which it follows
$$
\Gamma(s)=\bar{\Gamma}(s)+\int_{t_0}^s d\tau \int_{t_0}^\tau d\xi \ \big[ E(X(\xi),\xi)-E(Y(\xi),\xi) \big].
$$
By  (\ref{C1})
\begin{equation}
\begin{split}
\int_{t_0}^s d\tau\int_{t_0}^\tau d\xi \,& |E(X(\xi),\xi)-E(Y(\xi),\xi)|\leq 2C_2 P^{\frac43} Q^{\frac13}
\frac{|s-t_0|^2}{2}\leq 
 \\
&C_2 P^{\frac43} Q^{\frac13} \Delta  |s-t_0|\leq P^\gamma \frac{ |s-t_0|}{4}. 
\end{split}
\label{eq_app}
\end{equation}
 Hence,
\begin{equation}
 |\Gamma(s)|\geq |\bar{\Gamma}(s)|- P^\gamma \frac{|s-t_0|}{4}.
\label{z}
\end{equation}
Now we have:
$$
|\bar{\Gamma}(s)|^2=|\Gamma(t_0)|^2+2\Gamma(t_0)\cdot \dot{\Gamma}(t_0)(s-t_0)+|\dot{\Gamma}(t_0)|^2 |s-t_0|^2.
$$
We observe that $\Gamma(t_0) \cdot \dot{\Gamma}(t_0) (s-t_0)\geq 0.$  Indeed, if $t_0 \in(t', t'+\Delta)$ 
then $\dot{\Gamma}(t_0)=0$ while if $t_0=t'$ or $t_0=t'+\Delta$ the product $\Gamma(t_0) \cdot \dot{\Gamma}(t_0) (s-t_0)\geq 0$.
Hence
$$
|\bar{\Gamma}(s)|^2\geq |\dot{\Gamma}(t_0)|^2 |s-t_0|^2.
$$
By Lemma \ref{lemv3}, since $t_0\in [t', t'+\Delta]$ it is
$$
|\dot{\Gamma}(t_0)|\geq  \frac{P^\gamma}{2}
$$
hence
$$
|\bar{\Gamma}(s)|\geq  \frac{P^\gamma}{2} |s-t_0|
$$
and finally by (\ref{z}), 
$$
 |\Gamma(s)|\geq  \frac{P^\gamma}{4}|s-t_0|,
 $$
which concludes the proof.

\subsection*{Proof of  Lemma \ref{lemv3l}}

\begin{remark}
Before giving the proofs of Lemmas  
\ref{lemv3l}  and    \ref{lemrectl} we observe that it holds:
\begin{equation}
\sup_{t\in[0, T-\Delta_\ell]}\langle E\rangle_{\Delta_\ell}\leq 
\sup_{t\in [0, T-\Delta_{\ell -1}]} \langle E\rangle_{\Delta_{\ell-1}} \qquad \forall \ell \leq \bar{\ell}.\label{ave}
\end{equation}
In fact, 
$\Delta_\ell=\mathcal{G}\Delta_{\ell-1}$ with $\mathcal{G}$ given in (\ref{G_factor}), so that:
\begin{equation*}
[t,t+\Delta_{\ell}]=\bigcup_{i=1}^\mathcal{G}[t+(i-1)\Delta_{\ell-1}, t+i\Delta_{\ell-1}],
\end{equation*}
hence
\begin{equation}
\frac{1}{\Delta_{\ell}}\int_t^{t+\Delta_{\ell}}|E(X(s),s)|ds \leq \max_i \frac{1}{\Delta_{\ell-1}}\int_{t+(i-1)\Delta_{\ell-1}}^{t+i\Delta_{\ell-1}}|E(X(s),s)|ds,
\end{equation}
whence we get (\ref{ave}).

\label{rem_av}
\end{remark}

\bigskip

In order to prove Lemma \ref{lemv3l} 
we show now that Lemma \ref{lemv3} holds true also over a time interval $\Delta_\ell$, $2\leq\ell < \bar\ell$, 
under the assumption that 
estimate (\ref{av_ell}) at level $\ell - 1$ holds.  Indeed, by the use of (\ref{av_ell}) we get, for any $s\in[t, t+\Delta_\ell]$,  
\begin{equation}
\begin{split}
|&V(s) - W(s)| \leq |V(t)-W(t)| +\\& \int_{t}^{t+\Delta_\ell} \Big[ |E(X(\tau),\tau)| + |E(Y(\tau),\tau)| \Big] d\tau 
\leq \\&P^\gamma +  C   
\left[  P^{\frac43\gamma}Q^{\frac13} +
P^{\frac53-\eta} Q^\frac13 \log P + \frac{Q}{P^{\frac13+\gamma-\eta +(\ell-2)\delta}}  \right]\Delta_\ell:=\\&
P^\gamma +  C[a_1+a_2+a_3]\ \Delta_\ell,
  \end{split}
  \label{lemma_iter}
\end{equation}
where
$$
a_1=P^{\frac43\gamma}Q^{\frac13}, \qquad a_2=P^{\frac53-\eta} Q^\frac13 \log P,
\qquad a_3=\frac{Q}{P^{\frac13+\gamma-\eta +(\ell-2)\delta}}.
$$
Since $Q\leq R(T)^{\beta}\leq C P^{\beta},$ recalling that
$$
\Delta_\ell=\frac{\left[ \textnormal{intg}\left( P^\delta \right) \right]^{\ell-1}}{4 C_2  {P}^{\frac43-\gamma} Q^\frac13},
$$
we have:
\begin{equation}
a_1\Delta_\ell\leq  P^{\frac43 \gamma - \frac43 + \gamma + (\ell - 1)\delta}
\end{equation}
\begin{equation}
a_2\Delta_\ell\leq    P^{\frac13-\eta+\gamma+(\ell-1)\delta}    \log P
\end{equation}
and 
\begin{equation}
a_3\Delta_\ell\leq   P^{\frac23 \beta -\frac53 +\eta+\delta}.
\end{equation}
Being $\ell\leq\bar{\ell},$ from the definition of $ \bar{\ell}$ given in (\ref{elle_bar}) it follows that 
\begin{equation}
\beta-\frac13+\eta-\gamma-(\bar\ell-2)\delta \geq \frac23,
\end{equation}
from which
\begin{equation}
(\ell - 1)\delta \leq (\bar\ell-1)\delta \leq\beta - \frac13 - \gamma + \eta -\frac23 + \delta
\label{app_l}
\end{equation}
(note that the right hand side of  (\ref{app_l})  is positive for $\beta\geq\frac{6}{19}$).
Therefore, by the choices (\ref{delta4}),  (\ref{gamma_iter}) and (\ref{eta_iter}) made on the parameters, it follows
$$
a_i\Delta_\ell\leq  P^{\theta}  \qquad \theta <\gamma, \  \qquad i=1,2,3.
$$
In fact, for the term $a_2\Delta_\ell$, it is sufficient to prove
$$
\frac13-\eta+\gamma+(\ell-1)\delta < \gamma,
$$
which holds true since, by (\ref{app_l}) and  (\ref{gamma_iter}),
$$
\frac13-\eta+(\ell-1)\delta \leq \beta - \gamma-\frac23+\delta<0.
$$
Let us observe that the interval in (\ref{gamma_iter}) is well defined if
$$
\beta-\frac23+\delta < \frac{2-\beta}{4},
$$
which gives condition (\ref{delta4})
$$
0< \delta < \frac76-\frac54\beta.
$$
Let us consider the term $a_3\Delta_\ell$, for which it is sufficient to prove
$$
\frac23 \beta -\frac53 +\eta+\delta < \gamma,
$$
that is 
$$
\eta < \frac53+\gamma-\frac23\beta-\delta,
$$
which is implied by (\ref{eta_iter}). The interval in (\ref{eta_iter}) is well defined, as it follows by the condition on $\gamma$
(\ref{gamma_iter}).

\noindent  Finally we examine the term $a_1\Delta_\ell$, for which we require
$$
\frac43 \gamma - \frac43 + \gamma + (\ell - 1)\delta < \gamma,
$$
which holds true since
\begin{equation*}
\begin{split}
\frac43 \gamma - \frac43 +  (\ell - 1)\delta  < \, &\frac43 \gamma - \frac43 +
\beta - \frac13 - \gamma + \eta -\frac23 + \delta = \\
&\frac{\gamma}{3}-\frac73+\beta+\eta+\delta < 0,
\end{split}
\end{equation*}
by taking 
$$
\eta < \frac73 - \frac{\gamma}{3}-\beta-\delta,
$$
condition which is automatically fulfilled by (\ref{eta_iter}), since
$$
\frac53+\gamma-\frac23\beta-\delta < \frac73 - \frac{\gamma}{3}-\beta-\delta,
$$
as it is evident by using (\ref{gamma_iter}).

Hence, provided that $P$ is sufficiently large (as the constant $\widetilde C$ in (\ref{mv}) assures), we have: 
$$
C[a_1+a_2+a_3]\ \Delta_\ell\leq P^{\gamma}
$$
which proves the thesis.

We proceed analogously for the lower bound.

\subsection*{Proof of  Lemma \ref{lemrectl}}

We note that the same arguments used in the proof of Lemma \ref{lemrect} work also for Lemma  \ref{lemrectl},
considering the interval $[t, t+\Delta_\ell]$,  $\ell > 1$,
and  for the electric field
the estimate (\ref{av_ell}) at level $\ell - 1$.  
In fact, going back to  (\ref{eq_app}), we have for any $s\in[t', t'+\Delta_{\ell}]$
\begin{equation}
\int_{t_0}^s d\tau\int_{t_0}^\tau d\xi \, |E(X(\xi),\xi)-E(Y(\xi),\xi)|\leq 2 \langle E \rangle_{\Delta_{\ell-1}} \, \Delta_\ell
\int_{t_0}^s d\tau   \leq 
\frac{P^\gamma}{4}  |s-t_0|,
\end{equation}
where we treat the term $\langle E \rangle_{\Delta_{\ell-1}}  \Delta_\ell$ in the same way as in (\ref{lemma_iter}).
 Hereafter the proof proceeds in the same way as the proof of Lemma \ref{lemrect}.

\end{document}